\documentclass[dvipsnames]{article}

\usepackage{arxiv}
\usepackage{hyperref}
\usepackage{amsmath} 
\usepackage{amssymb}  
\usepackage{amsthm}
\usepackage{nicematrix}
\usepackage{graphicx}
\usepackage[ruled,vlined]{algorithm2e}
\usepackage{booktabs}
\usepackage{subcaption}

\usepackage{xcolor}  

\usepackage{amssymb}

\usepackage{amsmath}
\usepackage{mathdots}

\usepackage{mathtools}

\usepackage{tensor}

\usepackage{urwchancal}
\DeclareFontFamily{OT1}{pzc}{}
\DeclareFontShape{OT1}{pzc}{m}{it}{<-> s * [1.10] pzcmi7t}{}
\DeclareMathAlphabet{\mathpzc}{OT1}{pzc}{m}{it}

\usepackage{graphicx}
\usepackage{ifpdf}
\ifpdf
\usepackage{epstopdf}
\PrependGraphicsExtensions{.svg}
\fi

\newtheorem{theorem}{Theorem}[section]
\newtheorem{lemma}[theorem]{Lemma}

\newtheorem{remark}[theorem]{Remark}

\providecommand{\Order}{\mathbb{O}}

\providecommand{\R}{\mathbb{R}}
\providecommand{\E}{\mathbb{E}} 

\providecommand{\SO}{\mathbf{SO}}

\providecommand{\SE}{\mathbf{SE}}

\providecommand{\gothX}{\mathfrak{X}} 

\providecommand{\calC}{\mathcal{C}}

\providecommand{\calJ}{\mathcal{J}}

\providecommand{\calM}{\mathcal{M}}
\providecommand{\calN}{\mathcal{N}}

\providecommand{\calR}{\mathcal{R}}

\providecommand{\vecL}{\mathbb{L}}

\providecommand{\tT}{\mathrm{T}} 

\providecommand{\GP}{\mathbf{N}} 

\DeclareMathOperator{\Exp}{Exp}

\DeclareMathOperator{\tr}{tr}

\DeclareMathOperator{\diag}{diag}

\providecommand{\PT}{\mathbf{P}} 

\providecommand{\td}{\mathrm{d}}
\providecommand{\tD}{\mathrm{D}}

\providecommand{\ddt}{\frac{\td}{\td t}}

\usepackage{accents}
\usepackage{mathtools}
\makeatletter
\providecommand{\scirc}{%
    \hbox{\fontfamily{\rmdefault}\fontsize{0.4\dimexpr(\f@size pt)}{0}\selectfont{\raisebox{-0.52ex}[0ex][-0.52ex]{$\circ$}}}}

\makeatother

\makeatletter
\providecommand{\ucirc}{%
    \hbox{\fontfamily{\rmdefault}\fontsize{0.4\dimexpr(\f@size pt)}{0}\selectfont{\raisebox{0.0ex}[0ex][-0.52ex]{$\circ$}}}}

\makeatother

\mathchardef\mhyphen="2D

\providecommand{\Order}{\mathbf{O}} 
\providecommand{\order}{\mathbf{o}} 

\providecommand{\etal}{\textit{et al.~}}

\DeclareFontEncoding{LS1}{}{}
\DeclareFontSubstitution{LS1}{stix}{m}{n}
\DeclareSymbolFont{stixletters}{LS1}{stix}{m}{it}
\DeclareMathAccent{\cev}{\mathord}{stixletters}{"91}
\DeclareMathAccent{\vec}{\mathord}{stixletters}{"92}
\DeclareMathAccent{\vecev}{\mathord}{stixletters}{"95}

\providecommand{\xiop}{\check{\xi}_{k+1}}
\providecommand{\dds}{\frac{\td}{\mathrm{ds}}}

\newcommand{\bR}{\mathbf{R}}
\providecommand{\bv}{\mathbf{v}}
\providecommand{\bp}{\mathbf{p}}

\providecommand{\jp}[3]{J^{#1}_1({#2},#3)}
\providecommand{\jt}[3]{J^{#1}_2({#2},#3)}

\renewcommand{\Order}{\mathbf{O}}

\newcommand{\publicationdetails}{Under review}
\newcommand{\publicationversion}
{Preprint}

\title{\LARGE \bf The Geometry of Extended Kalman Filters on Manifolds\\ with Affine Connection}

\author{
    \href{https://orcid.org/0000-0001-7969-7039}{\includegraphics[scale=0.06]{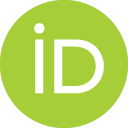}\hspace{1mm}
Yixiao Ge}
\\
    Systems Theory and Robotics Group \\
    School of Engineering \\
	Australian National University \\
    ACT, 2601, Australia \\
    \texttt{Yixiao.Ge@anu.edu.au} \\
\And    \href{https://orcid.org/0000-0003-4391-7014}{\includegraphics[scale=0.06]{orcid.png}\hspace{1mm}
Pieter van Goor}
\\
    Robotics and Mechatronics (RaM) Group \\
    EEMCS Faculty \\
	University of Twente \\
    Enschede, The Netherlands \\
    \texttt{p.c.h.vangoor@utwente.nl} \\
	\And	\href{https://orcid.org/0000-0002-7803-2868}{\includegraphics[scale=0.06]{orcid.png}\hspace{1mm}
    Robert Mahony}
\\
    Systems Theory and Robotics Group \\
    School of Engineering \\
	Australian National University \\
    ACT, 2601, Australia \\
	\texttt{Robert.Mahony@anu.edu.au} \\
}

\begin{document}
\maketitle

\begin{abstract}
The extended Kalman filter (EKF) has been the industry standard for state estimation problems over the past sixty years.
The classical formulation of the EKF is posed for nonlinear systems defined on global Euclidean spaces.
The design methodology is regularly applied to systems on smooth manifolds by choosing local coordinates, however, it is well known that this approach is not intrinsic to the manifold and performance depends heavily on choosing ``good''  coordinates.
In this paper, we propose an extended Kalman filter that is adapted to the specific geometry of the manifold in question.
We show that an affine connection and the concepts of parallel transport, torsion, and curvature are the key geometric structures that allow the formulation of a suitable family of intrinsic Gaussian-like distributions and provide the tools to understand how to propagate state estimates and fuse measurements.
This leads us to propose novel geometric modifications to the propagation and update steps of the EKF and revisit recent work on the geometry of the reset step.
The relative performance of the proposed geometric modifications are benchmarked against classical EKF and iterated EKF algorithms on a simplified inertial navigation system with direct pose measurements and no bias.
\end{abstract}

\section{Introduction}
\label{sec:introduction}

The extended Kalman filter (EKF) has been the industry standard non-linear state estimation algorithm for the past sixty years \cite{maybeck1982stochastic,Anderson2012,barfoot2017state}.
The original Kalman filter \cite{kalman1960new} was extended by Smith \etal \cite{smith1962application,mcgee1985discovery} to non-linear systems using local linearisations, introducing the now standard prediction-update form of the filter, and defining the reset step to ensure that the linearisations are computed at the most recent estimate of the state.
Since its creation, many modifications of the EKF have been proposed to handle system nonlinearity better.
The iterated EKF (ItEKF) \cite{gelb1974applied} repeatedly re-linearises the measurement function around temporary posteriors in order to minimise the linearisation error in the update step.
High-order EKFs were developed by considering the higher-order terms in the Taylor series when linearising the system functions \cite{jazwinski2007stochastic}.
The robust EKF (REKF) uses $\mathbf{H}_\infty$ techniques to design filter updates that are robust to higher-order terms in the Taylor series \cite{einicke1999robust}.
The unscented Kalman filter (UKF) \cite{julier2004unscented} approaches the problem from a more stochastic perspective and uses the unscented transformation to model the propagation of non-Gaussian distributions and estimate the joint state-measurement distribution in the update step.

Although the classical formulation of the EKF was posed on global Euclidean state space, the state space for the original application of EKF, to attitude estimation in the Apollo missions \cite{smith1962application}, was the special orthogonal group.
While any choice of local coordinates provides a representation in which the EKF can be implemented, from as early as the 1970s, authors were demonstrating the advantage of choosing local coordinate charts that encode geometric structure \cite{duncan1977some,lefferts1982kalman,ng1985nonlinear}.
A natural class of manifolds to consider are Riemannian manifolds \cite{duncan1977some,ng1985nonlinear} equipped with a Riemannian metric.
This structure is closely related to the $\boxplus$ (`boxplus') and $\boxminus$ (`boxminus') operators defined in \cite{hertzberg2013integrating} and used in implementations of the UKF \cite{clemens2016extended,Barrau_2017,hauberg2013unscented} and iterated EKF on manifolds \cite{bloesch2017iterated}.
The Levi-Civita connection, the unique torsion free connection that preserves the Riemannian metric, provides structure to study diffusion processes \cite{duncan1977some,said2013filtering} and to formulate sigma point constellations for the UKF \cite{hauberg2013unscented} in intrinsic coordinates.
A separate class of manifolds that comprise many core applications for the EKF are Lie-groups and homogeneous spaces.
Apart from the special orthogonal group $\SO(n)$ that admits a bi-invariant Riemannian metric, most Lie-groups and homogeneous spaces do not admit an intrinsic Riemannian metric \cite{nomizu1954invariant}.
However, the group structure itself has been extensively exploited in algorithm design.
The quaternion representation for attitude was used to define the multiplicative EKF (MEKF) \cite{lefferts1982kalman} exploiting a local group error \cite[Equation (118)]{lefferts1982kalman} that leads to a ``nice'' structure for linearisation of the attitude state~\cite{markley2003attitude} and is a form of error state Kalman filter \cite{1999_Roumeliotis}.
Lie-groups also have an intrinsic coordinate structure given by the exponential or one-parameter subgroup curves.
Chirikjian \etal~\cite{wang2006error,long2013banana} showed that for left invariant kinematics on Lie-groups the linearised propagation in the error state is exact.
Bonnabel \etal developed the theory of the invariant extended Kalman filter (IEKF)~\cite{Bonnabel_2007,Bonnabel_2009} that exploits this property for a class of systems on Lie-groups termed ``group affine''.
Mahony \etal~\cite{mahony2020equivariant} studied equivariant systems on homogeneous spaces leading to the Equivariant Filter (EqF) \cite{van2022equivariant,ge2022equivariant}.

In this work, we consider nonlinear systems where the state space and output space are both smooth manifolds that admit an affine connection.
An affine connection is the minimum geometric structure that captures the geometry of the space without requiring a Riemannian metric, a class of manifolds that we will term \emph{geometric manifolds}.
Geometric manifolds admit a local geodesic map that provides \emph{normal coordinates} around every point in the manifold.
Using these intrinsic coordinates, we extend the concept of concentrated Gaussian distribution on Lie-groups \cite{wolfe2011bayesian, Bourmaud_2015,ge2024geometric} to the case of geometric manifolds to provide a stochastic model for the information state of the filter.
With this formulation we can define parameters for the mean and covariance of the filter information state in normal coordinates on the manifold without requiring the associated stochastic concepts to be well defined for the true information state on the manifold.
The main result in the paper provides the tools to change coordinates for  extended concentrated Gaussians and shows that the formulae associated with these changes are intimately linked to the geometry of the manifold through the parallel transport, torsion, and curvature operators.
Once different extended concentrated Gaussians can be written in the same normal coordinates then the classical formulae for the extended Kalman filter can be applied.
The two key modifications to the filter are in the update and reset steps.
In the filter update, the covariance of the generative measurement noise model should be expressed in common coordinates with the state estimates.
The authors believe this is a novel contribution.
In the filter reset, the a posteriori state covariance estimate needs to be expressed in new coordinates associated with the updated state estimate, a result that is present in prior literature \cite{markley2003attitude,mueller2017covariance,gill2020full}.
Implementing the geometric insight leads to modifications that refine the classical EKF formulae, even accounting for using the intrinsic normal coordinates.
We provide simulations to demonstrate the advantage of the proposed modification, and note that the primary advantage is seen in the initial transient response of the filter.
Interestingly, the results show that both geometric modifications are required to generate significant performance advantages in the filter response, perhaps explaining why the literature on reset steps has not been adopted by the mainstream filter community to date.
A major secondary contribution of the paper is the derivation of easily computable approximations of the key Jacobian operators that are required to implement the geometric modifications.

This paper includes seven sections alongside the introduction and conclusion.
Section \ref{sec:prelim} defines the notation and discusses the mathematical tools used.
In Section \ref{sec:concentratedGaussian}, we present the concept of concentrated Gaussian distributions on geometric manifolds, as well as the associated mappings.
In Section  \ref{sec:geometric_ekf} we explain the EKF methodology on a manifold and propose novel geometric modifications in the filter dynamics.
We present an example of inertial navigation problem with direct pose measurements in Section \ref{sec:simulation}.
The Appendix provides the necessary background on the Jacobi field and how to use the associated ODE to derive the algebraic approximations of the Jacobians required in the filter implementation.

\section{Preliminaries}\label{sec:prelim}

For a comprehensive introduction to smooth manifolds and their geometric structures, the authors recommend \cite{lee2012smooth} and \cite{lee2018introduction}.

Let $\calM$ be a smooth manifold with dimension $m$.
For an arbitrary $\xi\in\calM$, the tangent space at $\xi$ is denoted $\tT_\xi\calM$.
The tangent bundle is denoted $\tT\calM$.
Let $\gothX(\calM)$ denote the space of smooth vector fields over $\calM$.
Let $C^\infty(\calM)$ denote the class of smooth functions on $\calM$.
Given a differentiable function between smooth manifolds $h:\calM\rightarrow\calN$, its derivative at $\xi_\circ$ is written as
\begin{align*}
    \tD_{\xi|\xi_\circ}h(\xi): \tT_{\xi_\circ}\calM\rightarrow \tT_{h(\xi_\circ)}\calN.
\end{align*}
The notation $\tD h:\tT\calM\rightarrow \tT\calN$ denotes the differential of $h$ with an implicit base point.
For $v \in \tT_\xi \calM$ then $\tD h[v] \in \tT_{h(\xi)} \calN$ is the evaluation of $\tD h$ at the point $\xi$.

\subsection{Geometric Manifolds}
A \emph{geometric manifold} is defined as a differentiable manifold equipped with an affine connection.

\subsubsection{Affine Connection}
An affine connection \cite{lee2018introduction} is an operator $\nabla:\gothX(\calM)\times\gothX(\calM)\rightarrow\gothX(\calM)$ written $\nabla(X,Y)\mapsto\nabla_X Y$, which gives a notion of directional derivative of a vector field $Y$ in direction $X$ on the manifold.
It satisfies the product rule of differentiation
\[
\nabla_X(fY)=f\nabla_X Y+ (\tD f [X]) Y.
\]
for all $f\in C^\infty(\calM)$ and $X,Y\in\gothX(\calM)$.
It is $\calC(\calM)$-linear in $X$, meaning
\[
\nabla_{f_1 X_1+f_2 X_2} Y = f_1 \nabla_{X_1}Y+f_2 \nabla_{X_2}Y,
\]
for all $f_1,f_2 \in C^\infty(\calM)$ and all $X_1,X_2,Y \in \gothX(\calM)$.
It is $\R$-linear in $Y$, meaning
\[
\nabla_X (a_1 Y_1+a_2 Y_2) = a_1 \nabla_X Y_1+a_2 \nabla_X Y_2,
\]
for all $a_1, a_2\in\R$ and all $X,Y_1,Y_2 \in \gothX(\calM)$.

\subsubsection{Geodesic and exponential mapping}
Consider a smooth manifold $\calM$ with an affine connection.
A curve $\gamma:I_\gamma\rightarrow\calM$ on an interval $I_\gamma \subset \R$ is called geodesic \cite{lee2018introduction} if
\begin{align*}
   \nabla_{\dot{\gamma}(t)}\dot{\gamma}(t) = 0
\end{align*}
for any $t\in I_\gamma$, where $I_\gamma$ is a maximal open interval in $\R$ containing 0 for which $\gamma(t)$ is uniquely defined.
For any $\hat{\xi} \in \calM$ and $v \in \tT_{\hat{\xi}}\calM$, there exists a time $t(v,\hat{\xi}) > 0$ (possibly infinite) and a unique maximal geodesic $\gamma_v:[0,t(v,\hat{\xi})) \rightarrow \calM$ that satisfies $\gamma_v(0)=\hat{\xi}$ and $\dot{\gamma}_v(0) = v$ \cite[Collorary 4.28]{lee2018introduction}.
The exponential mapping $\Exp_{\hat{\xi}}: W_{\hat{\xi}} \subset \tT_{\hat{\xi}}\calM\rightarrow\calM$ is defined by mapping each tangent vector $v \in W_{\hat{\xi}}$ to the value of its geodesic at $t=1$, that is,
\begin{align*}
   \Exp_{\hat{\xi}}(v) = \gamma_v(1)
\end{align*}
where $W_{\hat{\xi}}$ is the largest open subset of $\tT_{\hat{\xi}} \calM$ for which $\Exp_{\hat{\xi}}$ is a diffeomorphism.
Let $U_{\hat{\xi}} = \Exp_{\hat{\xi}} (W_{\hat{\xi}}) \subset \calM$ and note that $U_{\hat{\xi}}$ is open by construction.
Let $\imath_{\hat{\xi}} : \tT_{\hat{\xi}} \calM \to \R^m$ provide a locally smoothly varying linear isomorphism between $\tT_{\hat{\xi}} \calM$ and $\R^m$ for each $\hat{\xi}$.
Then the \emph{normal coordinates} on $\calM$ are the family of charts defined by
\begin{align}
\vartheta_{\hat{\xi}}:=  \Exp_{\hat{\xi}}^{-1} \circ\; \imath_{\hat{\xi}}: U_{\hat{\xi}} \to \R^m.
\label{eq:vartheta}
\end{align}

\begin{remark}
On a matrix Lie-group equipped with a Cartan-Schouten affine connection \cite{nomizu1954invariant} then the geometric $\Exp$ mapping is the same as the matrix exponential $\exp$ associated with the one-parameter Lie-subgroups.
\end{remark}

\subsubsection{Parallel transport}
A vector field $X$ is \emph{parallel} along $\gamma$ with respect to the connection $\nabla$ if
\begin{align}
   \nabla_{\dot{\gamma}(t)} X_{\gamma(t)} = 0
\label{eq:parallel_transport}
\end{align}
for all $t \in I_\gamma$.
Given any vector $X_{\gamma(0)} \in \tT_{\gamma(0)} \calM$ and smooth curve $\gamma(t)$ then there is a unique family of vectors $X_{\gamma(t)}$ that satisfy \eqref{eq:parallel_transport}.
This correspondence induces an invertible linear  map, termed \emph{parallel transport}, $\PT^{0,t}_{\gamma} : \tT_{\gamma(0)} \calM \to \tT_{\gamma(t)} \calM$ between tangent spaces by
\[
\PT^{0,t}_{\gamma} X_{\gamma(0)} := X_{\gamma(t)}
\]
for all $X_{\gamma(0)} \in \tT_{\gamma(0)} \calM$.
The inverse of parallel transport is parallel transport along the same geodesic but in reverse direction, $(\PT^{0,t}_{\gamma})^{-1} = \PT_{\gamma}^{t,0}$.

\subsubsection{Curvature}\label{sec:curvature}

The Riemann curvature tensor $\calR: \mathfrak{X}(M) \times \mathfrak{X}(M) \times \mathfrak{X}(M) \rightarrow \mathfrak{X}(M)$ is defined as
\[
    \calR(X, Y) Z=\nabla_X \nabla_Y Z-\nabla_Y \nabla_X Z-\nabla_{[X, Y]} Z,
\]
where $[X,Y]$ is the Lie bracket of vector fields.
Note that this definition does not require a Riemannian metric despite the similarity in name.
The Riemann curvature tensor is linear in $X$, $Y$ and $Z$ due to linearity of the affine connection and is a (3,1)-tensor.

\subsubsection{Partial pushforward maps of the exponential mapping}\label{sec:differential_exp}
Let $v_\xi \in \tT \calM$ denote an element of the tangent bundle with $v_\xi \in \tT_\xi \calM$. 
We will write $(\xi, v) = (\xi, \imath_\xi (v_\xi))$ to denote a local trivialisation of tangent vectors $v_\xi$ and their $\calM \times \R^m$ coordinates. 
The full exponential map, $(\xi, v) \mapsto \vartheta^{-1}_\xi(v) = 
\Exp_\xi (v_\xi)$, is defined on a neighborhood of $(\calM,0) \subset \tT \calM$ \cite{lee2018introduction}, and depends both on the argument $v \in \R^m$ and on the base point $\xi \in \calM$ in local coordinates. 
Let $\tD_\xi\vartheta^{-1}_\xi(v):\tT_\xi\calM\cong\R^m\to\tT_{\vartheta^{-1}_\xi(v)}\calM\cong\R^m$ and $\tD_v\vartheta^{-1}_\xi(v):\tT_v\tT_\xi\calM\cong\R^m\to\tT_{\vartheta^{-1}_\xi(v)}\calM\cong\R^m$ denote the partial pushforward maps of the exponential map with respect to the positional argument $\xi$ and the tangential argument $v_\xi$ respectively.
    For $w\in\R^m$, these maps are given by
    \begin{align*}
        \tD_\xi\vartheta^{-1}_\xi(v)[w] & :=  \left.\ddt\right|_{t=0}\vartheta^{-1}_{\vartheta^{-1}_\xi(tw)}(v),\\
        \tD_v\vartheta^{-1}_\xi(v)[w] & := \left.\ddt\right|_{t=0}\vartheta^{-1}_\xi(v+tw).
    \end{align*}

In the robotics community \cite{Chirikjian_2011}, the \emph{tangential} partial pushforward map $\tD_v\vartheta^{-1}_\xi(v)$ is referred to as the \emph{Jacobian} of the exponential map \cite{Chirikjian_2011} and is often written $J(\zeta)$, for $\zeta = \vartheta^{-1}_\xi(v) \in \calM$, 
with the base point implicit. 
Since we consider Jacobians with respect to both position $\xi$ and tangential $\zeta = \vartheta^{-1}_\xi(v)$ variation we will use $\jp{}{\xi}{\zeta}$ as a shorthand for $\tD_\xi\vartheta^{-1}_\xi(v)$ and $\jt{}{\xi}{\zeta}$ as a shorthand for $\tD_v\vartheta^{-1}_\xi(v)$.
The subscripts $1$ and $2$ correspond to differentiating with respect to the positional and tangential arguments respectively.
In Appendix~\ref{sec:approximation_dexp}, we use Jacobi field theory and curvature to provide algebraic approximations of the partial pushforward maps that overcomes the challenge of computing analytic expressions for these transcendental functions.

The left-inverse Jacobian $\jt{-1}{\xi}{\zeta}:\tT_\zeta\calM\to\tT_\xi\calM$ is given by
\begin{align*}
    \jt{-1}{\xi}{\zeta}[w] & = \tD_\zeta \vartheta_\xi(\zeta)\cdot\tD_\epsilon|_0\vartheta^{-1}_\zeta(\epsilon)[w] \\  
    & = \left.\ddt\right|_{t=0}\vartheta_\xi(\vartheta^{-1}_\zeta(tw)).
\end{align*}
It is straightforward to see that $\jt{-1}{\xi}{\zeta}\cdot \jt{}{\xi}{\zeta}[w] = w$.

\begin{remark}
    Although there may exist specific instances where the equality holds, in the general case, for an arbitrary $w\in\tT_{\zeta}\calM$, it does not necessarily follow that $\jt{}{\xi}{\zeta}\cdot \jt{}{\zeta}{\xi}[w] = w$ \cite[Chapter 4, Lemma 3.5]{lang2012fundamentals} since the derivatives are evaluated at different coordinates.
\end{remark}

\begin{lemma}\label{lemma:dxlog}
    Given $\xi,\zeta\in\calM$ and $w\in\R^m$, the differential of the logarithmic map $\vartheta_\xi(\zeta)$ with respect to the positional argument $\xi$ satisfies
    \begin{align*}
        \tD_\xi \vartheta_\xi(\zeta)[w] &= -\left(\tD_v\vartheta^{-1}_\xi(v)\right)^{-1}\cdot \tD_\xi\vartheta^{-1}_\xi(v)[w] \\
                            &= -\jt{-1}{\xi}{\zeta} \cdot \jp{}{\xi}{\zeta}[w],
    \end{align*}
    with $v = \vartheta_\xi(\zeta)$.
\end{lemma}
\begin{proof}
    Let $\xi(t):[0,1]\to\calM$ be a smooth curve on $\calM$ with $\xi(0)=\xi$ and $\dot{\xi}(0)=w$, such that the whole curve lies in the normal neighborhood of $\zeta$, i.e., $\lambda(t):=\vartheta_{\xi(t)}(\zeta)$ is a smooth vector field along the curve that is always well-defined.

    Notice that one has $\zeta = \vartheta^{-1}_{\xi(t)}(\vartheta_{\xi(t)}(\zeta))$.
    Since $\zeta$ is fixed, differentiating both sides with respect to $t$ gives
    \begin{align*}
        0 = \left.\tD_v\right|_{\lambda(t)}\vartheta^{-1}_{\xi(t)}(v)[\nabla_{\dot{\xi}(t)}\lambda(t)] + \left.\tD_z\right|_{\xi(t)}\vartheta^{-1}_z(\lambda(t))[\dot{\xi}(t)].
    \end{align*}
    Rearranging the terms yields
    \begin{align}\label{eq:nabla_log}
        \nabla_{\dot{\xi}(t)}\lambda(t) = -\left(\left.\tD_v\right|_{\lambda(t)}\vartheta^{-1}_{\xi(t)}(v)\right)^{-1}\cdot\left.\tD_z\right|_{\xi(t)}\vartheta^{-1}_z(\lambda(t))[\dot{\xi}(t)].
    \end{align}
    Setting $t = 0$ one has $\tD_\xi \vartheta_\xi(\zeta)[w] = \nabla_{\dot{\xi}(0)}\lambda(0)$ and the result follows.
\end{proof}

\subsection{The $\boxplus/\boxminus$ Operators}

Building on established convention we will use the $\boxplus$ and $\boxminus$ operator notation introduced in \cite{hertzberg2013integrating} to model small $\R^m$ `perturbations' acting on $\calM$.
We will use the normal coordinates \eqref{eq:vartheta} to define \emph{normal box plus} $\boxplus: \calM \times \R^m \rightarrow\calM$  and \emph{normal box minus} $\boxminus:\calM\times\calM\rightarrow\R^m$ operators
by
\begin{align}
   \xi \boxplus u & = \vartheta_\xi^{-1} (u), \label{eq:boxplus} \\
   \zeta \boxminus \xi & = \vartheta_{\xi}(\zeta), \label{eq:boxminus}
\end{align}
for all $\xi \in\calM$, $\zeta \in U_\xi$ and $u \in \imath_\xi(W_\xi)$.
Both the $\boxplus$ and $\boxminus$ operators are associated with geodesic curves on the manifold and in this sense are the natural generalisation of straight line interpolation on Euclidean space.

\begin{remark}
The box-plus and box-minus notation is fundamentally a choice of local coordinates; \eqref{eq:boxminus} is a chart while \eqref{eq:boxplus} is the associated parametrization.
These operators are used in prior work \cite{aufframework,hertzberg2013integrating} to define distances between points, essentially using the Euclidean distance in local coordinates as a distance measure on the manifold.
On a Riemannian manifold the geodesic is the natural generalisation of a distance minimizing curve and hence on such manifolds the \emph{normal box minus} is the natural choice to capture this geometric concept.
In the present work, we do not restrict to only Riemannian geometry and cannot use the concept of distance although the construction still provides an intrinsic coordinate structure.  \hfill $\Box$
\end{remark}

It is straightforward to verify that the proposed operators satisfy
\begin{subequations} \label{eq:boxplus_axioms}
   \begin{align}
      \xi \boxplus 0 &= \xi, \\
      \xi \boxplus (\zeta \boxminus \xi) &= \zeta, \\
      (\xi \boxplus u) \boxminus \xi &= u,
   \end{align}
\end{subequations}
the first three of the four requirements of the original definition proposed in \cite[Def.~1]{hertzberg2013integrating}.
The fourth axiom in \cite[Def.~1]{hertzberg2013integrating} requires that
\begin{gather}
|(\hat{\xi} \boxplus \delta_1) \boxminus (\hat{\xi} \boxplus \delta_2)|^2 \leq |\delta_1 - \delta_2|^2.
\label{eq:axiom4}
\end{gather}
This will not hold on a general manifold.
To see this consider the particular case where the affine connection $\nabla$ is the Levi-Civita connection associated with a Riemannian metric, a subset of the set of geometric manifolds.
In this case, the Ricci curvature $\text{Ric} : \tT_\xi \calM \times \tT_\xi \calM \to \R$ is defined and one has \cite{muller1997closed}
\[
|(\hat{\xi} \boxplus \delta_1) \boxminus (\hat{\xi} \boxplus \delta_2)|^2
=
|\delta_1 - \delta_2|^2
-\frac{1}{3} \text{Ric}(\delta_1, \delta_2) + \order(|\delta|^2).
\]
For manifolds with non-negative curvature; that is, $\text{Ric} \geq 0$ is positive semi-definite, this is always true.
However, for any manifold with negative curvature, \eqref{eq:axiom4} fails even locally.

\begin{remark}
The fourth axiom in \cite[Def.~1]{hertzberg2013integrating} was used to prove properties of the mean of the true information state on the manifold $\calM$ associated with properties of the mean of distributions defined in the $\R^m$ chart.
In general manifolds the concept of mean and covariance on the manifold are unclear and such a correspondence will not be possible.
\end{remark}

\section{Concentrated Gaussian Distributions} \label{sec:concentratedGaussian}
In this section, we revise the concept of a concentrated Gaussian distribution \cite{wang2006error} to geometric manifolds, that is, smooth manifolds that admit affine connections.
The main results in the section provide formula for coordinate transformations of concentrated Gaussian distributions.

A global volume measure can be defined on any smooth manifold with an affine connection using the local Borel measures induced by normal coordinates and a partition of unity construction.
The information state of the system lies in the class of all probability distributions on the manifold that are integrable with respect to this measure.
Such distributions are extremely general and, on arbitrary geometric manifold, concepts such as the mean and covariance are not well-defined.
An extended Kalman filter algorithm does not directly approximate the true system state probability distribution.
Rather, an EKF algorithm searches within a set of \emph{Gaussian-like} probability distributions, parameterised by mean and covariance parameters, for the distribution that is ``closest'' to the true system information state in an informal information theoretic sense.
A crucial insight is that the mean and covariance distribution parameters used in the filter formulation only have meaning in the sense of providing a finite dimensional parametrisation for the filter algorithm, \emph{they do not need to correspond to a stochastic concept of mean or covariance of the system information state distribution}.
In other words, it is not necessary that the mean and covariance parameters used as state in an EKF correspond to statistics of the filter distribution or converge to statistics of the true distribution as long as the distribution generated by the filter parametrisation is close stochastically to the true distribution.


\subsection{Concentrated Gaussian Definition}\label{sub:stochastic_model}

In the remainder of the paper we assume that both the system state space $\calM$ and output space $\calN$ admit affine connections and that we work with normal coordinates \eqref{eq:vartheta}
\begin{align}
\vartheta_{\hat{\xi}} &: \calM \rightarrow \R^m \\
\varphi_{\hat{y}}  & : \calN \rightarrow \R^n
\end{align}
We approximate a general distribution $p : \calM \to \R_+$ around $\hat{\xi} \in \calM$ by a \emph{concentrated Gaussian distribution} (function of $\xi \in \calM$) \cite{wang2006error}
\begin{align}
\GP_{\hat{\xi}}(\mu,\Sigma) := \lambda^{-1}
\exp(-\frac{1}{2}(\vartheta_{\hat{\xi}}(\xi)-\mu)^\top\Sigma^{-1}(\vartheta_{\hat{\xi}}(\xi)-\mu)),
\label{eq:ConcentratedGaussian}
\end{align}
where
\[
\lambda := \left| \int_{U_{\hat{\xi}}}
\exp(-\frac{1}{2}(\vartheta_{\hat{\xi}}(\xi)-\mu)^\top\Sigma^{-1}(\vartheta_{\hat{\xi}}(\xi)-\mu))
 \td \xi \right|
\]
is a normalizing factor, $\mu\in \R^m$ is a mean vector parameter and $\Sigma\in\mathbb{S}_+(m)$ is a positive-definite symmetric $m\times m$ covariance matrix parameter.
Note that the support for the distribution $\GP_{\hat{\xi}}(\mu,\Sigma)$ is contained in the open set $U_{\hat{\xi}} \subset \calM$.
Within this set, the distribution in normal coordinates $x = \vartheta_{\hat{\xi}}(\xi)$ looks like a trimmed Gaussian, and the first and second order statistics $\mu$ and $\Sigma$ have stochastic interpretations.



\subsection{Coordinate Changes}
\label{sec:coordinate_transform}
\begin{figure}[!htb]
    \centering
    \includegraphics[width=0.5\linewidth]{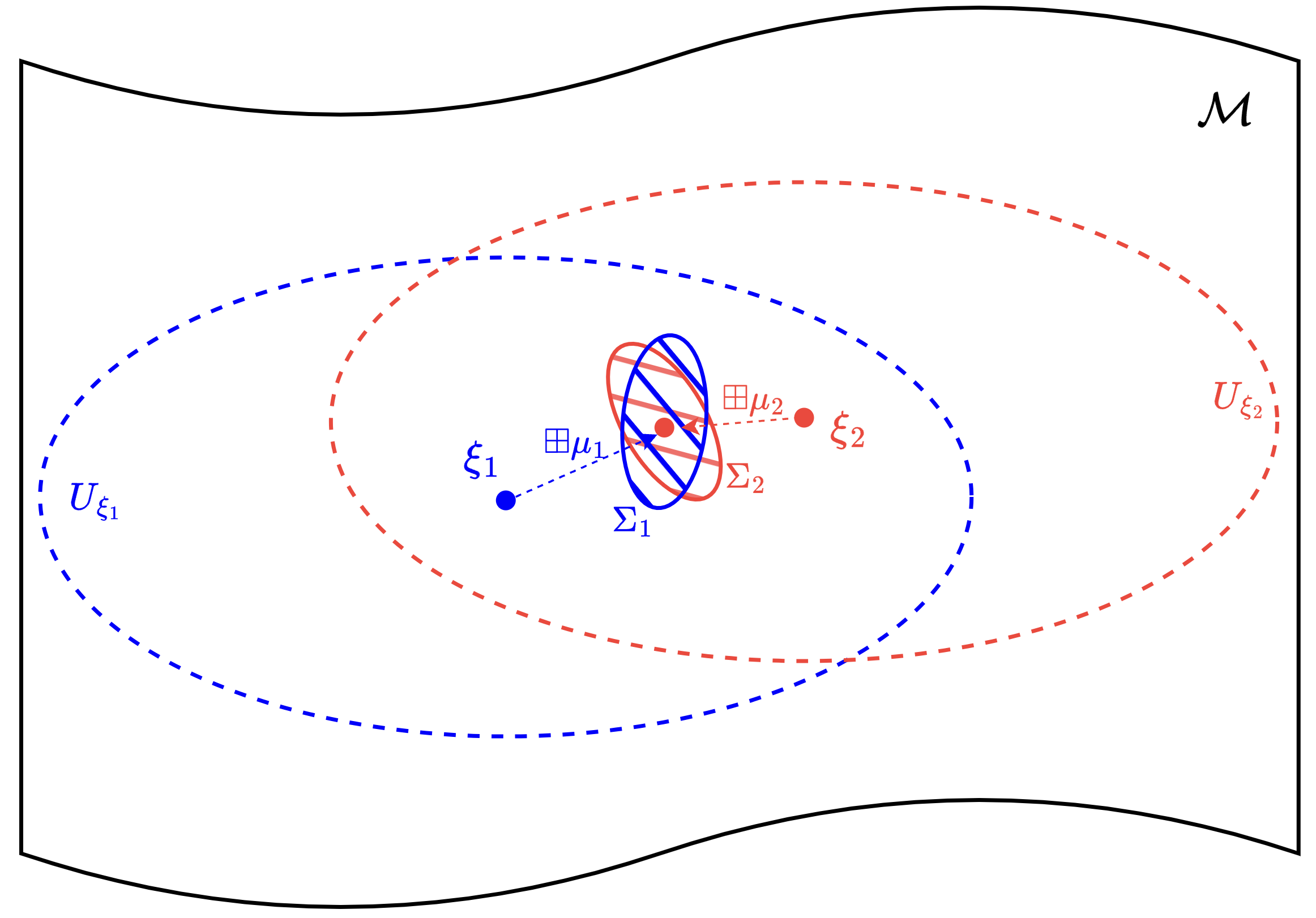}
    \caption{
The relationship between two concentrated Gaussian distributions expressed in different normal coordinates on $\calM$.
The hatched ellipse in blue represents the concentrated Gaussian $\xi\sim\GP_{\xi_1}(\mu_1,\Sigma_1)$, defined in $U_{\xi_1}$. The hatched ellipse in red represents the concentrated Gaussian $\xi\sim\GP_{\xi_2}(\mu_2,\Sigma_2)$, defined in $U_{\xi_2}$, where $\zeta = \xi_1\boxplus\mu_1 = \xi_2\boxplus\mu_2$.}
\label{fig:coordinate_transform}
\end{figure}

An important property of Definition \eqref{eq:ConcentratedGaussian} is that it formally allows for distributions in local coordinates with non-zero means.
One of the key technical result in this paper is to provide an expression for the
``closest'' concentrated Gaussian expressed in a desired set of normal coordinates to a concentrated Gaussian defined in another different set of normal coordinates.
This result is key requirement in order to apply the standard update and reset equations in Kalman filtering as well as having insight for the propagation step.

To make the problem clear, consider a known distribution in local coordinates centered at $\xi_1$ with mean $\mu_1$ and covariance $\Sigma_1$ as shown in blue in Figure~\ref{fig:coordinate_transform}.
The goal is to find the concentrated Gaussian (shown in red in Figure~\ref{fig:coordinate_transform}) defined in coordinates centered at $\xi_2$
that most closely approximates the original distribution in a stochastic sense.
That is, we are looking for the concentrated Gaussian defined in coordinates centered at $\xi_2$
\begin{align}\label{eq:pdf_1}
    \GP_{\xi_2}&(\mu_2,\Sigma_2):=\nonumber \\
    &\lambda_2^{-1} \exp((-\frac{1}{2}(\vartheta_{\xi_2}(\xi)-\mu_2)^\top\Sigma_2^{-1}(\vartheta_{\xi_2}(\xi)-\mu_2)),
\end{align}
(where $\lambda_2$ is a normalizing factor) that is closest to $\GP_{\xi_1}(\mu_1,\Sigma_1)$.
Note that in the case when $\mu_1=0$ this corresponds to the update modification in an EKF algorithm as proposed in Section~\ref{sec:update}.
When $\mu_2=0$ this corresponds to the reset step.
One may imagine $\Sigma_1 = \Sigma_2$ would be the correct choice, however, that will only be the case if the space $\calM$ is flat (with zero curvature).

\begin{theorem}\label{theorem:coordinate_transform}
    Given a random variable $x\in\R^m$ represented by a Gaussian distribution $x\sim\GP(0,\Sigma)$, let $\theta:\R^m\to\R^m$ be an analytic nonlinear function that satisfies $\theta(0)=0$.
    Assume the covariance $\Sigma_{\theta} = \E[\theta(x)\theta(x)^\top]$ of $\theta(x)$ is well-defined.
    Then $\Sigma_{\theta}$ can be computed by
    \begin{align}
    \Sigma_{\theta} = (\tD_z|_0\theta(z))\;\Sigma\;(\tD_z|_0\theta(z))^\top + \Order(\tr(\Sigma^2)).
    \end{align}
\end{theorem}

\begin{proof}
    Write the Taylor expansion of $\theta(x)$ around $x=0$.
    Since $\theta(0)=0$, one has
    \[
    \theta(x) = \tD_z|_0\theta(z)\,[x] + \frac{1}{2}\tD^2_z|_0\theta(z)[x,x] + \dots
    \]
    Given the definition of the covariance, one has $\Sigma = \E[xx^\top]$ and $\Sigma_{\theta} = \E[\theta(x)\theta(x)^\top]$.
    Substituting the expansion into the definition of the covariance $\Sigma_{\theta}$, one has
    \begin{align*}
    \Sigma_{\theta} &= \E[\theta(x)\theta(x)^\top]\\
    & = \E\big[(\tD_z|_0\theta(z)[x] + \frac{1}{2}\tD^2_z|_0\theta(z)[x,x] + \dots)\\
    &\hspace{2cm}(\tD_z|_0\theta(z)[x] + \frac{1}{2}\tD^2_z|_0\theta(z)[x,x] + \dots)^\top\big].
    \end{align*}
    The second moment can be derived:
    \begin{align*}
        &\quad\E\big[(\tD_z|_0\theta(z)[x])(\tD_z|_0\theta(z)[x] )^\top\big]\\
        &= \E\big[\tD_z|_0\theta(z)[x][x^\top]\tD_z|_0\theta(z)^\top\big]\\
        &= (\tD_z|_0\theta(z))\E[xx^\top](\tD_z|_0\theta(z))^\top\\
        &= (\tD_z|_0\theta(z))\;\Sigma\;(\tD_z|_0\theta(z))^\top.
    \end{align*}
    Note that when taking the expectation under $x \sim \mathcal{N}(0,\Sigma)$, any monomial of x of odd total degree has zero expectation, hence the third moment and all the odd moments of $x$ vanish.
    The fourth moment  $\E[x_i\,x_j\,x_k\,x_\ell]$ is given by
    \[
    \E[\,x_i\,x_j\,x_k\,x_\ell\,]\;=\;\Sigma_{i j}\;\Sigma_{k \ell}\;+\;\Sigma_{i k}\;\Sigma_{j \ell}\;+\;\Sigma_{i \ell}\;\Sigma_{j k}.
    \]
    It is of order $\tr(\Sigma^2)$ as each entry is a sum of products of $\Sigma$'s two entries, and all the even higher moments are at least of order $\tr(\Sigma^2)$.
    This completes the proof.
\end{proof}

\begin{lemma}\label{lemma:standard_linearisation}
    Given an arbitrary concentrated Gaussian distribution $p(\xi) = \GP_{\xi_1}(\mu_1,\Sigma_1)$ on $\calM$, define $\xi^\diamond = \vartheta^{-1}_{\xi_1}(\mu_1)$.
    For any $\xi_2\in\calM$ in the neighborhood $U_{\xi_1}$ around $\xi_1$, then the concentrated Gaussian distribution $\GP_{\xi_2}(\mu_2,\Sigma_2)$ for parameters
    \begin{align}
    \mu_2 & = \vartheta_{\xi_2}\vartheta^{-1}_{\xi_1}(\mu_1) \label{eq:mu_2}\\
    \Sigma_2 & = \jt{-1}{\xi_2}{\xi^\diamond}\;\jt{}{\xi_1}{\xi^\diamond} \;\Sigma_1 \;\jt{\top}{\xi_1}{\xi^\diamond}\;\jt{-\top}{\xi_2}{\xi^\diamond},\label{eq:Sigma_2}
    \end{align}
    achieves a minimal Kullback-Leibler divergence up to second-order error $\Order(\tr({\Sigma_1}^2))$.
\end{lemma}
\begin{proof}
    Define $q(\xi) = \GP_{\xi_2}(\mu_2,\Sigma_2)$ for general parameters $\mu_2\in\R^m$ and $\Sigma_2 \in S_+(m)$.
    The Kullback-Leibler divergence between $p(\xi)$ and $q(\xi)$ is given by
    \begin{align*}
    {\text{KL}}(p||q)&= \E_p[\log(p)-\log(q)]\\
    & = C_p - \E_p[\log(q)]\\
    & = C_p +\frac{m}{2}\log(2\pi)+ \frac{1}{2}\log\det(\Sigma_2)\\
    &\qquad+\frac{1}{2}\E_p[(\vartheta_{\xi_2}(\xi)-\mu_2)^\top \Sigma_2^{-1}(\vartheta_{\xi_2}(\xi)-\mu_2)],
    \end{align*}
    where $C_p$ is the negative entropy of $p(\xi)$ and $m$ is the dimension of $\Sigma$.
    We show that the differential of ${\text{KL}}(p||q)$ at \eqref{eq:mu_2} and \eqref{eq:Sigma_2} is zero.

    The derivative of $ {\text{KL}}(p||q)$ with respect to $\mu_2$ in an arbitrary direction $w \in \R^m$ is given by
    \begin{align*}
    \tD_{\mu_2} &\text{KL}(p||q)[w]\\
    & = \frac{1}{2}\tD_{\mu_2}\E_p[(\vartheta_{\xi_2}(\xi)-\mu_2)^\top \Sigma_2^{-1}(\vartheta_{\xi_2}(\xi)-\mu_2)][w]. \\
    & = - \frac{1}{2}\E_p[w^\top \Sigma_2^{-1}(\vartheta_{\xi_2}(\xi)-\mu_2)].
    \end{align*}
    Substituting, for \eqref{eq:mu_2} one finds that
    $\tD_{\mu_2}\text{KL}(p||q)[w] \equiv 0$ for all $w \in \R^m$.

    The derivative of $ {\text{KL}}(p||q)$ with respect to $\Sigma_2$ in an arbitrary direction $W \in S(m) \equiv \tT_\Sigma S_+(m)$ is given by
    \begin{align*}
    \tD_{\Sigma_2}&\text{KL}(p||q)[W] = \frac{1}{2}\tD_{\Sigma_2}\log\det({\Sigma_2})[W]\\
    &\qquad+\frac{1}{2}\tD_{\Sigma_2}\E_p[(\vartheta_{\xi_2}(\xi)-\mu_2)^\top \Sigma_2^{-1}(\vartheta_{\xi_2}(\xi)-\mu_2)][W]
    \end{align*}
    where we use the Frobenius inner product on matrix space to compute the direction derivative in the direction $W \in S(m)$.
    Solving the first term yields
    \begin{align*}
    \tD_{\Sigma_2}\log\det({\Sigma_2})[W]&=\frac{1}{\det(\Sigma_2)}\det(\Sigma_2)\tr(\Sigma_2^{-1}u)\\
    &=\tr(\Sigma_2^{-1}W).
    \end{align*}
    Solving the second term yields
    \begin{align*}
    &\tD_{\Sigma_2}\E_p[(\vartheta_{\xi_2}(\xi)-\mu_2)^\top \Sigma_2^{-1}(\vartheta_{\xi_2}(\xi)-\mu_2)][W]\\
    &=\E_p[(\vartheta_{\xi_2}(\xi)-\mu_2)^\top \Sigma_2^{-1}W \Sigma_2^{-1}(\vartheta_{\xi_2}(\xi)-\mu_2)]\\
    &=\E_p[\tr((\vartheta_{\xi_2}(\xi)-\mu_2)^\top \Sigma_2^{-1}W \Sigma_2^{-1}(\vartheta_{\xi_2}(\xi)-\mu_2))]\\
    &=\tr( \Sigma_2^{-1} \E_p[(\vartheta_{\xi_2}(\xi)-\mu_2)(\vartheta_{\xi_2}(\xi)-\mu_2)^\top] \Sigma_2^{-1}W).
    \end{align*}
    Combining terms one gets
    \begin{align}
    &\tD_{\Sigma_2}\text{KL}(p||q)[W] = \notag \\
    &\vspace{1cm}\frac{1}{2}\tr((\Sigma_2^{-1}+
   \Sigma_2^{-1} \E_p[(\vartheta_{\xi_2}(\xi)-\mu_2)(\vartheta_{\xi_2}(\xi)-\mu_2)^\top] \Sigma_2^{-1})W).\label{eq:DSigmaKL}
    \end{align}
    Setting \eqref{eq:DSigmaKL} to zero yields
    \begin{align}
        \Sigma_2^* = \E_p[(\vartheta_{\xi_2}(\xi)-\mu_2)(\vartheta_{\xi_2}(\xi)-\mu_2)^\top], \label{eq:SigmaOpt}
    \end{align}
    where $\Sigma_2^*$ is the optimal value of $\Sigma_2$ that minimises $ {\text{KL}}(p||q)$.

Recall that by definition, $\Sigma_1 = \E_p[(\vartheta_{\xi_1}(\xi)-\mu_1)(\vartheta_{\xi_1}(\xi)-\mu_1)^\top]$.
Define
\[
x_p = \vartheta_{\xi_1}(\xi)-\mu_1 \quad \text{and}\quad x_q = \vartheta_{\xi_2}(\xi)-\mu_2,
\]
for some arbitrary $\xi\in U_{\xi_1}$ sampled from $p(\xi)$.
One has
\begin{align}\label{eq:p2q}
x_q = \vartheta_{\xi_2}\vartheta^{-1}_{\xi_1}(x_p+\mu_1)-\mu_2.
\end{align}
Taking the Taylor series of \eqref{eq:p2q} evaluated at $x_p=0$ yields:
\begin{align*}
x_q &= 0 + \tD \vartheta_{\xi_2}(\xi^\diamond)\tD\vartheta^{-1}_{\xi_1}(\mu_1) x_p + \Order(\lvert x_p \rvert^2) \\
&= \jt{-1}{\xi_2}{\xi^\diamond}\;\jt{}{\xi_1}{\xi^\diamond} x_p+ \Order(\lvert x_p \rvert^2).
\end{align*}
Using Theorem~\ref{theorem:coordinate_transform} one can compute
\begin{align*}
    \Sigma_2^* & = \E_p[x_q x_q^\top]\\
    & = \jt{-1}{\xi_2}{\xi^\diamond}\;\jt{}{\xi_1}{\xi^\diamond} \;\Sigma_1 \;\jt{\top}{\xi_1}{\xi^\diamond}\;\jt{-\top}{\xi_2}{\xi^\diamond} \\
    & \qquad\qquad\qquad + \Order(\tr(\Sigma_1^2)).
\end{align*}
Substituting into \eqref{eq:DSigmaKL} it follows that
$\tD_{\Sigma_2}\text{KL}(p||q)[u] \equiv 0$ for all $W \in S(m)$.

\end{proof}


\section{Geometric EKF on smooth manifolds}\label{sec:geometric_ekf}
In this section, we extend the conventional EKF design to systems on smooth manifolds and integrate the geometric structure into the filter construction.
We derive the filter in the error-state formulation \cite{lefferts1982kalman,1999_Roumeliotis,sola2017quaternion}, that considers the propagation of the information state of an error $\varepsilon$ between the true state $\xi$ and the nominal state $\hat{\xi}$.

\subsection{System Model}
\label{sec:problem_description}
We consider a nonlinear discrete-time system on a smooth manifold $\calM$ with system function
\begin{align}\label{eq:system_twonoise}
    \xi_{k+1} = F&(\xi_k,u_{k+1} + \kappa^I_{k+1}) \boxplus \kappa^P_{k+1}, \notag\\
    &\kappa^I_{k+1}\sim\GP(0,Q^I_{k+1}),\notag\\
    &\kappa^P_{k+1}\sim\GP_{F(\xi_k,u_{k+1} + \kappa^I_{k+1})}(0,Q^P_{k+1}),
\end{align}
where $u_{k+1}\in\vecL$ is the vector input.
There are two types of process noise present in typical systems; the input noise $\kappa^I_{k+1}$, and the state noise $\kappa^P_{k+1}$.
The input noise is modelled as a Gaussian distribution on the linear input space $\vecL$, while the state noise is a concentrated Gaussian distribution on $\calM$.

We consider a simplified model with linearised input noise,
\begin{align}
\xi_{k+1} = F&(\xi_k,u_{k+1}) \boxplus (\kappa^P_{k+1} + B_{k+1}\kappa^I_{k+1}),
\end{align}
where $B_{k+1} = \frac{\partial F(\xi_k, u_{k+1})}{\partial u_{k+1}}$.
Combining the noise terms into a single process noise term one has
\begin{align}\label{eq:system}
\xi_{k+1} = F&(\xi_k,u_{k+1}) \boxplus \kappa_{k+1},\\
        &\kappa_{k+1}\sim\GP_{ F(\xi_k,u_{k+1}) }(0,Q_{k+1}),\notag
\end{align}
with $Q_{k+1} = Q^P_{k+1} + B_{k+1}Q^I_{k+1}B_{k+1}^\top$.

To understand the geometry of the noise term, define a virtual state $\xi^\diamond_{k+1} = F(\xi_k,u_{k+1})$ that represents the propagated state without noise.
Note that the object $\xi_{k+1}^\diamond$ is virtual and not measureable in practice.
The noise process $\kappa_{k+1}$ in \eqref{eq:system} is formally defined in the normal coordinates around $\xi^\diamond_{k+1}$ and the true state $\xi_{k+1}$ is obtained by $\boxplus$ adding a realisation of the noise process \eqref{eq:system} to $\xi^\diamond_{k+1}$.

The configuration output
\begin{align}\label{eq:configuration}
    y_{k+1} = h(\xi_{k+1})\boxplus \nu_{k+1},\quad \nu_{k+1}\sim\GP(0,R_{k+1}),
\end{align}
is given by a function $h: \calM\rightarrow\calN$, where $\calN$ is a smooth manifold termed the \emph{output space}.
The disturbance $\nu_{k+1}$ is modelled as a concentrated Gaussian in normal coordinates around $h(\xi_{k+1})$ on $\calN$.

For both processing and measurement noise we assume that $R_k$ and $Q_k$ are known covariance matrices.

\subsection{Error State}\label{sec:error_state}
We first construct a \emph{local} error in normal coordinates $\varepsilon_k\in\tT_{\hat{\xi}_k}\calM$, given by
\begin{align}
    \varepsilon_k = \xi_k \boxminus \hat{\xi}_{k|k},
\end{align}
where $\xi_k$ and $\hat{\xi}_{k|k}$ are the true and estimated states.
Here the notation $\hat{\xi}_{k|k}$ indicates the estimated state has fused all data up to index $k$, including inputs $u_k$ and measurements $y_k$.
For any geometric manifold with affine connection, this construction is always locally possible.
If the state space admits symmetry, the error state can be defined globally using the group structure \cite{Mahony_2022}.
In such case, the local error is the local linearisation of this construction \cite{Barrau_2017}\cite{van2022equivariant}.

The information state of the error is
\[
\epsilon_k \sim \GP(0,\Sigma_{k|k})
\]
where the state estimate $\hat{\xi}_{k|k}$ is implicit in the definition of the error term.
The corresponding information state of the filter, that approximates the true distribution of the system state on $\calM$, is a concentrated Gaussian distribution expressed in normal coordinates, given by
\begin{align*}
    \xi_k \sim \GP_{\hat{\xi}_{k|k}}(0,\Sigma_{k|k}).
\end{align*}

\subsection{Propagation}\label{sec:propagation}
The propagation step involves the propagation of the state estimate $\hat{\xi}_{k|k}$, which also acts as the reference point in the information state of the filter, using the full nonlinear model of the system
\begin{align}
    \hat{\xi}_{k+1|k} = F(\hat{\xi}_{k|k}, u_{k+1}).
\end{align}

The propagation of the covariance estimate requires the linearisation of the local error update equation.
We define the predicted error as
\begin{align}
    \varepsilon_{k+1|k}  := \xi_{k+1} \boxminus \hat{\xi}_{k+1|k}.
\end{align}

\begin{lemma}\label{lemma:propagation}
The linearised dynamics of $\varepsilon_{k+1|k}$ is given by
    \begin{align}
        &\varepsilon_{k+1|k} =  A_{k+1}\varepsilon_{k|k} + \jt{-1}{\hat{\xi}_{k+1|k}}{\xi^\diamond_{k+1}}\kappa_{k+1} \nonumber\\
        & \hspace{4.5cm} + \Order(\lvert \varepsilon_{k|k}, \kappa_{k+1} \rvert^2),
    \end{align}
where $A_{k+1}$ is given by
    \begin{align}
        A_{k+1} := \tD \vartheta_{\hat{\xi}_{k+1|k}}(\hat{\xi}_{k+1|k}) \cdot \tD F_{u_{k+1}}(\hat{\xi}_{k|k})\cdot \tD \vartheta^{-1}_{\hat{\xi}_{k|k}}(0).
    \end{align}
\end{lemma}

\begin{proof}
    The predicted error can be written
    \begin{align}
        \varepsilon_{k+1|k} &= \xi_{k+1}\boxminus\hat{\xi}_{k+1|k}, \notag  \\
        &= (F(\xi_{k}, u_{k+1})\boxplus \kappa_{k+1} )\boxminus F(\hat{\xi}_k, u_{k+1}), \notag  \\
                &= (F(\hat{\xi}_k\boxplus \varepsilon_{k|k}, u_{k+1}) \boxplus \kappa_{k+1} )\boxminus F(\hat{\xi}_k, u_{k+1}). \label{eq:predicterror_dynm_todo}
    \end{align}

Substitude the $\boxplus/\boxminus$ operators with the normal coordinate chart $\vartheta$, one has
\begin{align}
    \varepsilon_{k+1|k}&= \vartheta_{\hat{\xi}_{k+1|k}} \vartheta^{-1}_{F_{u_{k+1}} (\vartheta^{-1}_{\hat{\xi}_{k|k}}(\varepsilon_{k|k}))}(\kappa_{k+1}),\nonumber\\
    &=\vartheta_{\hat{\xi}_{k+1|k}} (F_{u_{k+1}} (\vartheta^{-1}_{\hat{\xi}_{k|k}}(\varepsilon_{k|k}))\boxplus\kappa_{k+1}). \label{eq:prop_true}
\end{align}
Ideally, one would like to have the error dynamics in the form of
\begin{align}
    \varepsilon_{k+1|k}\approx \vartheta_{\hat{\xi}_{k+1|k}} (F_{u_{k+1}} (\vartheta^{-1}_{\hat{\xi}_{k|k}}(\varepsilon_{k|k}))) + \kappa_{k+1}, \label{eq:prop_wrong_todo}
\end{align}
such that the noise process $\kappa_{k+1}$ is linear and independent from the error state.
This is true on Euclidean spaces since the $\boxplus$ and $\boxminus$ are just vector addition and subtraction.
On a geometric manifold we will appeal to Lemma \ref{lemma:standard_linearisation} to get an approximation of the modified noise process $\kappa^\diamond_{k+1}\sim\GP_{\hat{\xi}_{k+1|k}}(0,Q^\diamond_{k+1})$ in the normal coordinate around $\hat{\xi}_{k+1|k}$.
The new noise covariance is given by
\begin{align}
    Q^\diamond_{k+1} = \jt{-1}{\hat{\xi}_{k+1|k}}{\xi^\diamond_{k+1}}\; Q_{k+1}\; \jt{-\top}{\hat{\xi}_{k+1|k}}{\xi^\diamond_{k+1}}.
\end{align}

Rewriting the error dynamics \eqref{eq:prop_true} one obtains
\begin{align}
    \varepsilon_{k+1|k}= \vartheta_{\hat{\xi}_{k+1|k}} (F_{u_{k+1}} (\vartheta^{-1}_{\hat{\xi}_{k|k}}(\varepsilon_{k|k}))) + \kappa^\diamond_{k+1}. \label{eq:prop_correct}
\end{align}
The formula for $A_{k+1}$ follows by applying the chain rule of differentiation to \eqref{eq:prop_correct} and evaluating at $\varepsilon_{k|k} = 0$.
\end{proof}

\begin{figure}[!htb]
    \centering
    \includegraphics[width=0.5\linewidth]{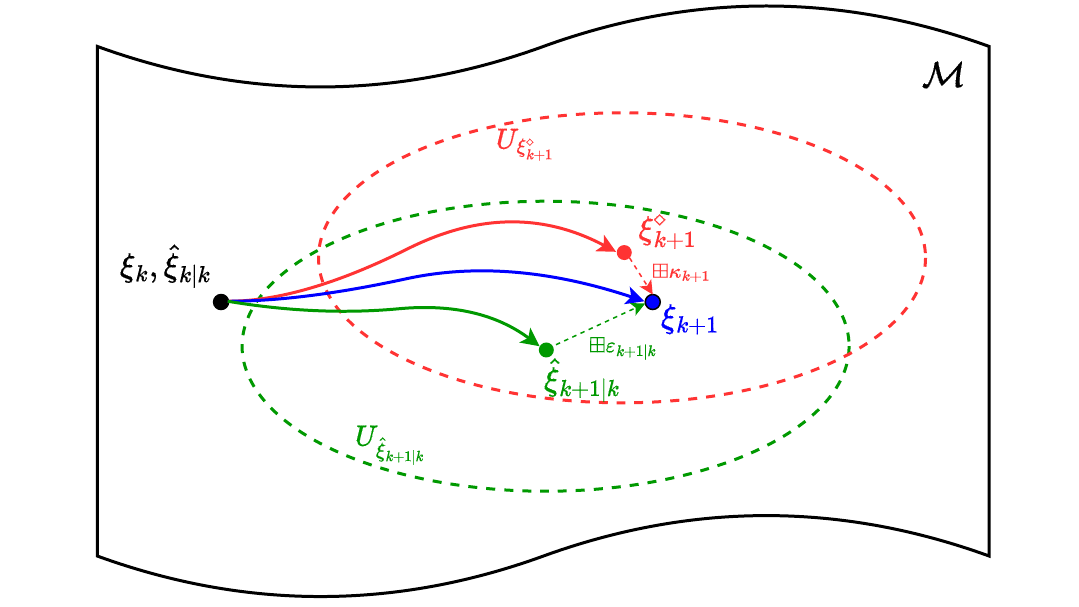}
    \caption{Demonstration of the filter propagation step. The true state $\xi_{k+1}$ and the nominal state $\hat{\xi}_{k+1|k}$ are in blue and green respectively. $\xi^\diamond_{k+1}$ is the ideal state which is propagated without the noise. The red and green dashed ellipses represent the domains of normal coordinates around $\xi^\diamond_{k+1}$ and $\hat{\xi}_{k+1|k}$. The red dashed arrow is the process noise $\kappa_{k+1}$ which is defined in $U_{\xi^\diamond_{k+1}}$. The green dashed arrow is the local error $\varepsilon_{k+1|k}$ which is defined in $U_{\hat{\xi}_{k+1|k}}$.}
    \label{fig:propagation}
 \end{figure}

The new information state of the filter after the propagation step becomes
\begin{align*}
    \xi_{k+1} \sim \GP_{\hat{\xi}_{k+1|k}}(0,\Sigma_{k+1|k}),
\end{align*}
where the updated covariance is given by (Lemma~\ref{lemma:propagation})
\begin{align*}
    &\Sigma_{k+1|k} = A_{k+1}\Sigma_{k|k}A_{k+1}^\top + Q^\diamond_{k+1}.
\end{align*}
This covariance update involves two parts of coordinate transforms.
The first transformation maps the covariance of the prior $\Sigma_{k|k}$ from $\hat{\xi}_{k|k}$ to $\hat{\xi}_{k+1|k}$.
This translation is handled by the state matrix $A_{k+1}$ implicitly.
The second transform rewrites the process noise covariance $Q_{k+1}$ located at $\xi^{\diamond}_{k+1}$ to $Q_{k+1}^+$ located at $\hat{\xi}_{k+1|k}$.

In practice $\xi^\diamond_{k+1}$ is not available and the Jacobian $ J_{\hat{\xi}_{k+1|k}}(\xi^\diamond_{k+1})$ is unknown.
At the moment, the authors do not have a way to overcome this difficulty and the best approximation that we are aware of is to set $Q_{k+1}$ to approximate $Q^\diamond_{k+1}$.
\begin{align}
    &\Sigma_{k+1|k} = A_{k+1}\Sigma_{k|k}A_{k+1}^\top + Q_{k+1}.
    \label{eq:Sigma_k+1|k}
\end{align}

\subsection{Update}\label{sec:update}
The update step is where the filter fuses the measurements $y_{k+1}\in\calN$ with the predicted prior $\xi_{k+1} \sim \GP_{\hat{\xi}_{k+1|k}}(0,\Sigma_{k+1|k})$ through a Bayesian fusion.
The noise model for the measurements can be written as
\[
    y_{k+1}\sim\GP_{h(\xi_{k+1})}(0, R_{k+1}),
\]
which is a concentrated Gaussian distribution on the normal coordinates around $h(\xi_{k+1})\in\calN$.

Applying the Bayes law the update step can be interpreted as solving for the parameters $\mu^+_{k+1|k}$ and $\Sigma^+_{k+1|k}$ that best approximate the posterior distribution.
Taking the log-likelihood, this can be written as finding
\begin{align}
    &\lVert \vartheta_{\hat{\xi}_{k+1|k}}(\xi_{k+1}) \rVert^2_{\Sigma^{-1}_{k+1|k}} + \lVert \varphi_{h(\xi_{k+1})}(y_{k+1}) \rVert^2_{R^{-1}_{k+1}} \nonumber\\
    &\hspace{2.5cm}\approx \lVert \vartheta_{\hat{\xi}_{k+1|k}}(\xi_{k+1})-\mu^+_{k+1|k} \rVert^2_{\Sigma_{k+1|k}^{+^{-1}}}.\label{eq:update_goal}
\end{align}

Define the local error $\varepsilon_{k+1}:= \vartheta_{\hat{\xi}_{k+1|k}}(\xi_{k+1}) \in \tT_{\hat{\xi}_{k+1|k}} \calM$, one has
\begin{align*}
    \eta(\varepsilon_{k+1}):=\varphi_{h(\xi_{k+1})}(y_{k+1}) = \varphi_{h(\vartheta^{-1}_{\hat{\xi}_{k+1|k}}(\varepsilon_{k+1}))}(y_{k+1}).
\end{align*}
Let $\hat{y}_{k+1}:=h(\hat{\xi}_{k+1|k})$ denote the predicted measurement.
Taking the Taylor expansion of $\eta(\varepsilon_{k+1})$ around $\varepsilon_{k+1}=0$ yields
\begin{align}\label{eq:measurement_update_taylor}
    &\eta(\varepsilon_{k+1}) \approx \varphi_{\hat{y}_{k+1}}(y_{k+1}) \nonumber\\
    &\hspace{0.3cm} + \left.\tD_z\right|_{\hat{y}_{k+1}} \varphi_z(y_{k+1}) \left.\tD_\zeta\right|_{\hat{\xi}_{k+1|k}} h(\zeta) \left.\tD_\epsilon\right|_{0} \vartheta^{-1}_{\hat{\xi}_{k+1|k}}(\epsilon) [\varepsilon_{k+1}].
\end{align}

Using the result in Lemma \ref{lemma:dxlog} and that $\left.\tD_\epsilon\right|_{0} \vartheta^{-1}_{\hat{\xi}_{k+1|k}}(\epsilon)$ is the identity at $\epsilon=0$, one can rewrite \eqref{eq:measurement_update_taylor} as
\begin{align}\label{eq:measurement_update_taylor2}
    &\eta(\varepsilon_{k+1}) \approx \varphi_{\hat{y}_{k+1}}(y_{k+1})\nonumber\\
     &\hspace{1cm}- \jt{-1}{\hat{y}_{k+1}}{y_{k+1}}\,\jp{}{\hat{y}_{k+1}}{y_{k+1}} \, C_{\hat{\xi}_{k+1|k}}\,\varepsilon_{k+1},
\end{align}
with $C_{\hat{\xi}_{k+1|k}} = \left.\tD_\zeta\right|_{\hat{\xi}_{k+1|k}} h(\zeta)$ the output Jacobian matrix.
Substitude \eqref{eq:measurement_update_taylor2} into \eqref{eq:update_goal} and solve for $\Sigma^+_{k+1|k}$ by matching terms yields
\begin{align}\label{eq:update_sigma}
    \Sigma^+_{k+1|k} = \left(\Sigma^{-1}_{k+1|k} + C_{\hat{\xi}_{k+1|k}}^\top R_\dagger^{-1} C_{\hat{\xi}_{k+1|k}}\right)^{-1},
\end{align}
with
\begin{align*} 
R_\dagger & = \jp{-1}{\hat{y}_{k+1}}{y_{k+1}}\, \jt{}{\hat{y}_{k+1}}{y_{k+1}}\, \\
 & \hspace{3cm} R\,\jt{\top}{\hat{y}_{k+1}}{y_{k+1}}\, \jp{-\top}{\hat{y}_{k+1}}{y_{k+1}}.
\end{align*} 
Using the matrix inversion lemma, it is straightforward to show that \eqref{eq:update_sigma} is equivalent to the usual Kalman filter covariance update equation,
\[\Sigma^+_{k+1|k}=(I-K_{k+1}C_{\hat{\xi}_{k+1|k}})\Sigma_{k+1|k},\]
where $K_{k+1}=\Sigma_{k+1|k} C_{\hat{\xi}_{k+1|k}}^\top(C_{\hat{\xi}_{k+1|k}}\Sigma_{k+1|k} C_{\hat{\xi}_{k+1|k}}^\top+R_\dagger)^{-1}$ is the Kalman gain.

Using \eqref{eq:measurement_update_taylor2}, one can also solve for the mean $\mu^+_{k+1|k}$ by matching terms in \eqref{eq:update_goal}, which yields
\begin{align}\label{eq:update_mu}
    \mu^+_{k+1|k} = -K_{k+1} \cdot \varphi_{\hat{y}_{k+1}}y_{k+1}.
\end{align}

Note that the Kalman update is performed in the normal coordinates around $\hat{\xi}_{k+1|k}$.
Hence, the posterior is also expressed in the same coordinate, that is, with the reference point remaining at $\hat{\xi}_{k+1|k}$.
The concentrated Gaussian distribution that is associated with the information state of the filter can be written as $\GP_{\hat{\xi}_{k+1|k}}(\mu^+_{k+1|k},\Sigma^+_{k+1|k})$.

\subsection{Reset}\label{sec:reset}
There have been several works on the covariance reset in error-state Kalman filters.
It was first mentioned by Markley \cite{markley2003attitude} in the context of multiplicative EKF, recently generalised by Muller \etal \cite{mueller2017covariance}\cite{gill2020full}.
The authors proposed a covariance reset step using parallel transport in \cite{Mahony_2022}\cite{ge2022equivariant} for filtering on homogeneous spaces, which was extended in \cite{ge2024geometric} to include the curvature tensor.
The same concept can be extended onto a smooth manifold.

\begin{figure}[!htb]
    \centering
    \includegraphics[width=0.5\linewidth]{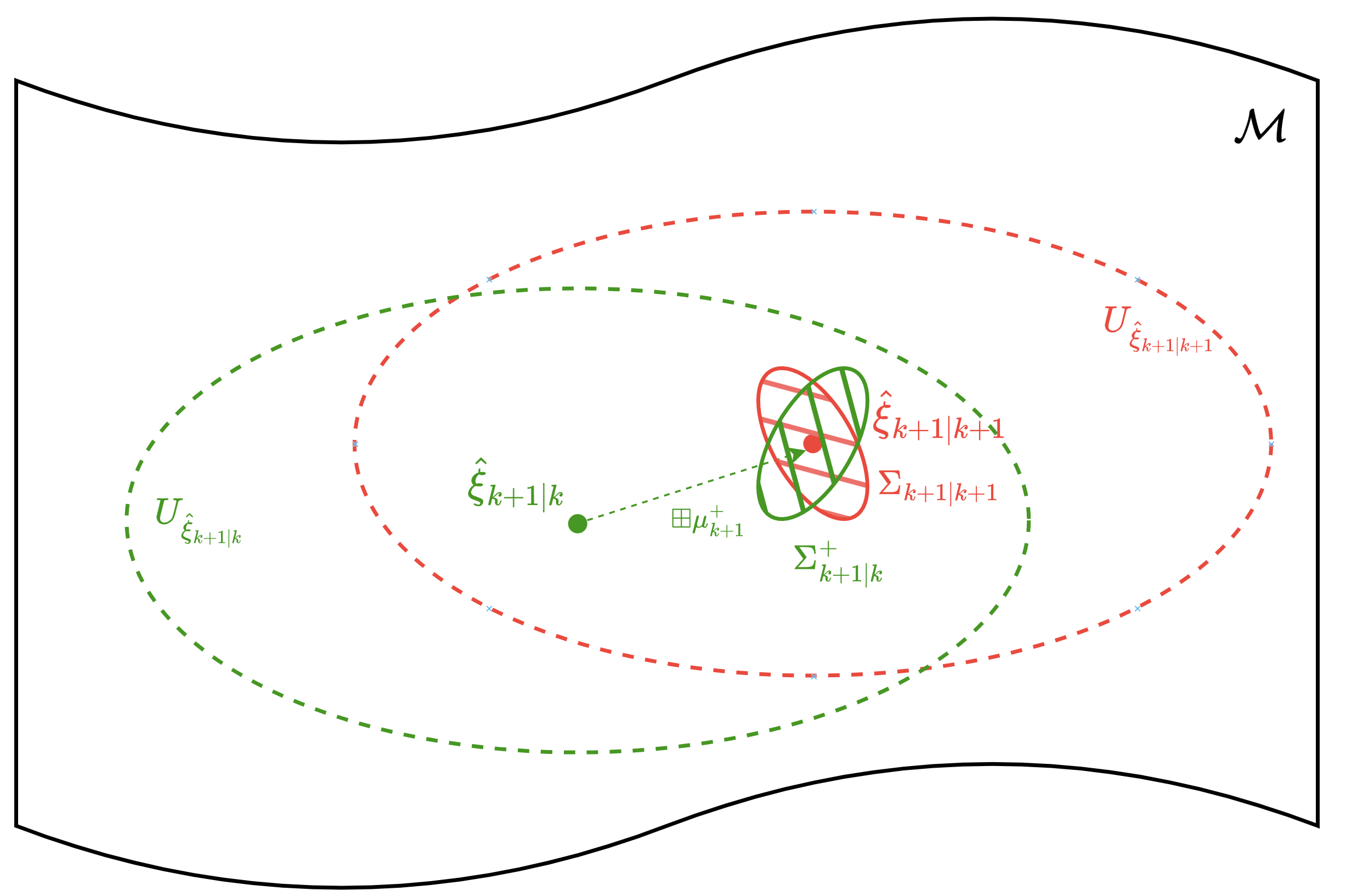}
    \caption{Demonstration of the filter reset step. The hatched ellipse in green represents the updated covariance $\Sigma^+_{k+1|k}$ which is expressed in normal coordinate around $\hat{\xi}_{k+1|k}$ with a non-zero mean $\mu^+_{k+1|k}$. The hatched ellipse in red is the covariance after reset $\Sigma_{k+1|k+1}$. It is expressed in the normal coordinate around $\hat{\xi}_{k+1|k+1}$ with a zero mean.}
    \label{fig:reset}
 \end{figure}

The information state of the filter after the update step is a concentrated Gaussian distribution in normal coordinates around $\hat{\xi}_{k+1|k}$, with non-zero mean $\mu^+_{k+1|k}$ and the updated covariance estimate $\Sigma^+_{k+1|k}$.
At the next time step $t=k+2$ in the filter, the linearisation will be evaluated around $\hat{\xi}_{k+1|k+1}$.
That is, the propagation outlined in Section~\ref{sec:propagation} requires a distribution in the normal coordinate around $\hat{\xi}_{k+1|k+1}$.
The goal of the reset step is to find $(\hat{\xi}_{k+1|k+1}, \Sigma^+_{k+1|k+1})$ that satisfy (see Fig.~\ref{fig:reset})
\begin{align}\label{eq:reset}
    \xi_{k+1} \sim \GP_{\hat{\xi}_{k+1|k}}(\mu^+_{k+1|k},\Sigma^+_{k+1|k})\approx\GP_{\hat{\xi}_{k+1|k+1}}(0,\Sigma_{k+1|k+1}).
 \end{align}
This step is very much in the original spirit of Smith's contribution \cite{smith1962application,mcgee1985discovery}.

The reset of the reference point is trivial, since one can get the updated state estimate $\hat{\xi}_{k+1|k+1}$ by computing
\[
    \hat{\xi}_{k+1|k+1} = \hat{\xi}_{k+1|k} \boxplus \mu_{k+1},
\]
however, the reset of the covariance estimate involves a coordinate transform, which is given by
\begin{align}
    \Sigma_{k+1|k+1} & = J_{\hat{\xi}_{k+1|k}}(\hat{\xi}_{k+1|k+1})\;\Sigma^+_{k+1|k+1}\;J_{\hat{\xi}_{k+1|k}}^\top(\hat{\xi}_{k+1|k+1})\label{eq:Sigma_reset}.
\end{align}

Unlike the geometric modification in the propagation and update steps, the information required in the reset step is always available in practice, hence \eqref{eq:Sigma_reset} can be solved explicitly.

Algorithm \ref{alg:geometric_EKF} summarises the proposed EKF methodology with each step outlined.

\begin{algorithm}[h]
    \caption{Geometric EKF on manifold}
    \KwIn{Initial state estimate $\hat{\xi}_0$ and covariance $\Sigma_0$.}
    \KwOut{Sequence of state and covariance estimates $\{ \hat{\xi}_k,\Sigma_k \}$.}
    \For{$k = 0, 1, 2, \ldots$}{
        \vspace{2mm}
        \textbf{\{Propagation step\}}\\
        \vspace{1mm}
        Define process noise covariance $Q_{k+1}$\\
        \vspace{2mm}
        $A_{k+1}\leftarrow\tD \vartheta_{\hat{\xi}_{k+1|k}}(\hat{\xi}_{k+1|k}) \cdot \tD F_{u_{k+1}}(\hat{\xi}_{k|k})\cdot \tD \vartheta^{-1}_{\hat{\xi}_{k|k}}(0)$\\
        $\hat{\xi}_{k+1}\leftarrow F(\hat{\xi}_{k}, u_{k+1})$\\
        $\Sigma_{k+1|k} \leftarrow A_{k+1}\Sigma_{k|k}A_{k+1}^\top + Q_{k+1}$\\
        \vspace{2mm}
        \If{Measurement $y_{k+1}$ available}
        {
            \vspace{2mm}
            \textbf{\{Update step\}}\\
            \vspace{1mm}
            Define measurement noise covariance $R_{k+1}$\\
            Compute the maps $\jt{-1}{\hat{y}_{k+1}}{y_{k+1}}, \jp{}{\hat{y}_{k+1}}{y_{k+1}}$ \\
            \vspace{2mm}
            $C_{\hat{\xi}_{k+1|k}}\leftarrow \tD h(\hat{\xi}_{k+1_k})$\\
            \vspace{3mm}
            $ R_\dagger \leftarrow  \jp{-1}{\hat{y}_{k+1}}{y_{k+1}} \jt{}{\hat{y}_{k+1}}{y_{k+1}} $\\
            $\hspace{2cm}R\jt{\top}{\hat{y}_{k+1}}{y_{k+1}}\jp{-\top}{\hat{y}_{k+1}}{y_{k+1}}$\\
            \vspace{3mm}
            $ K_{k+1} \leftarrow \Sigma_{k+1|k} C_{\hat{\xi}_{k+1|k}}^\top(C_{\hat{\xi}_{k+1|k}}\Sigma_{k+1|k} C_{\hat{\xi}_{k+1|k}}^\top+R_\dagger)^{-1}$\\
            \vspace{2mm}
            $\Sigma^+_{k+1|k} \leftarrow (I-K_{k+1}C_{\hat{\xi}_{k+1|k}})\Sigma_{k+1|k}$\\
            $\mu^+_{k+1|k}\leftarrow -K_{k+1} \cdot \varphi_{h(\hat{\xi}_{k+1|k})}(y_{k+1})$\\
            \vspace{1mm}
            \textbf{\{Reset step\}}\\
            \vspace{1mm}
            $\hat{\xi}_{k+1|k+1} \leftarrow \vartheta^{-1}_{\hat{\xi}_{k+1|k}}(\mu^+_{k+1|k})$\\
            Compute the map $\jt{}{\hat{\xi}_{k+1|k}}{\hat{\xi}_{k+1|k+1}}$\\
            $\Sigma_{k+1|k+1} \leftarrow \jt{}{\hat{\xi}_{k+1|k}}{\hat{\xi}_{k+1|k+1}}\Sigma^+_{k+1|k} \jt{\top}{\hat{\xi}_{k+1|k}}{\hat{\xi}_{k+1|k+1}}$\\
        }
    }
    \label{alg:geometric_EKF}

\end{algorithm}

\section{Iterated Extended Kalman Filters}\label{sec:iterated_EKF}

Linearisation error is the major issue in designing high performance filters using the EKF formulation \cite{huang2007convergence}.
The EKF filter marginalises out and discards all previous measurements $\{u_1, \ldots, u_k\}$, and $\{y_1, \ldots, y_k\}$ by storing the information in this data in the parameters $(\hat{\xi}_{k|k}, \Sigma_{k|k})$ of the information state.
Any linearisation error in the matrices $A_k$, $B_k$ or $C_k$, or the Jacobians $J_{\hat{\xi}_{k|k-1}}$, $J_{\hat{y}}$, and $T_{\hat{y}}$ get ``baked into'' the covariance parameter $\Sigma_{k|k}$ of the information state during the update step of the EKF.
This can be particularly harmful during the transient phase when large linearisation error can corrupt the covariance estimate leading to poor filter consistency and significantly degradinging robustness and filter convergence.
For high performance filters, a common approach to address this issue is to iterate the update step in the filter for each prediction step, re-linearizing the system equations repeatedly, until the linearisation error is minimised.
This algorithm is termed the iterated EKF \cite{bell1993iterated} algorithm.
In this section, we modify the proposed EKF methodology in Sec.~\ref{sec:update} to allow for an arbitrary linearisation point.
Once this is available, it is straightforward to pose a geometric iterated EKF algorithm.

\subsection{Choosing the linearisation point}\label{sec:linearisation_point}

Consider the question of fusing a measurement $y_{k+1} \sim \GP_{h(\xi_{k+1})}(0,R_{k+1})$ with a prior $\GP_{\hat{\xi}_{k+1|k}}(0,\Sigma_{k+1|k})$ where the linearisation point  $\xiop\in\calM$ used for the update step is not the predicted prior state-estimate $\hat{\xi}_{k+1|k}$.
The first step is to express the prior distribution, that is defined as a concentrated Gaussian $\GP_{\hat{\xi}_{k+1|k}}(0,\Sigma_{k+1|k})$ at the reference point $\hat{\xi}_{k+1|k}$, around the new linearisation point $\xiop$.

Define $\varepsilon_{k+1}:=\vartheta_{\xiop}(\xi_{k+1})$, one has $\vartheta_{\hat{\xi}_{k+1|k}}(\xi_{k+1}) = \vartheta_{\hat{\xi}_{k+1|k}}(\vartheta_{\xiop}^{-1}(\varepsilon_{k+1}))$.
Taking the Taylor expansion with respect to $\varepsilon$ around $\varepsilon=\vartheta_{\xiop}(\hat{\xi}_{k+1|k})$ yields
\begin{align*}
&\vartheta_{\hat{\xi}_{k+1|k}}(\xi_{k+1})\\
&\hspace{1cm}\approx \left.\tD_v\right|_{\vartheta_{\xiop}(\hat{\xi}_{k+1|k})}\vartheta^{-1}_{\xiop}(v)[\varepsilon_{k+1}-\vartheta_{\xiop}(\hat{\xi}_{k+1|k})]\\
&\hspace{1cm}= \jt{}{\xiop}{\hat{\xi}_{k+1|k}}\left(\varepsilon_{k+1}-\vartheta_{\xiop}(\hat{\xi}_{k+1|k})\right).
\end{align*}
Applying Lemma~\ref{lemma:standard_linearisation}, one can transform the prior distribution into the normal coordinate around $\xiop$, that is,
\begin{align}\label{eq:prior_linearisation}
    \lVert \vartheta_{\hat{\xi}_{k+1|k}}(\xi_{k+1}) \rVert^2_{\Sigma_{k+1|k}^{-1}} \approx \lVert \varepsilon_{k+1} - \vartheta_{\xiop}(\hat{\xi}_{k+1|k}) \rVert^2_{{\Sigma^{\dagger\;\;\;{-1}}_{k+1|k}}},
\end{align}
with $\Sigma^\dagger_{k+1|k} = \jt{-1}{\xiop}{\hat{\xi}_{k+1|k}}\;\Sigma_{k+1|k}\; \jt{-\top}{\xiop}{\hat{\xi}_{k+1|k}}$.

Similarly, the measurement likelihood can be written in coordinates $\eta(\varepsilon_{k+1}):=\varphi_{h(\xi_{k+1})}(y_{k+1}) = \varphi_{h(\vartheta_{\xiop}^{-1}(\varepsilon))}(y_{k+1})$.
Taking the Taylor expansion of $\eta(\varepsilon_{k+1})$ around $\varepsilon_{k+1}=0$ yields
\begin{align}
    \eta(\varepsilon_{k+1}) &\approx \varphi_{\check{y}_{k+1}}(y_{k+1}) + \left.\tD_z\right|_{\check{y}_{k+1}} \varphi_z(y_{k+1}) \nonumber\\
    &\quad \cdot \left.\tD_\zeta\right|_{\check{\xi}_{k+1}} h(\zeta) \cdot \left.\tD_\epsilon\right|_{0} \vartheta^{-1}_{\check{\xi}_{k+1}}(\epsilon) [\varepsilon_{k+1}],\label{eq:measurement_update_taylor_iterate}
\end{align}
where $\check{y}_{k+1}:=h(\xiop)$ is the linearisation point on $\calN$.
Using Lemma~\ref{lemma:dxlog}, one can rewrite \eqref{eq:measurement_update_taylor_iterate} as
\begin{align}\label{eq:measurement_update_taylor2_iterate}
    &\eta(\varepsilon_{k+1}) \approx \varphi_{\check{y}_{k+1}}(y_{k+1}) \nonumber\\
    & \hspace{1.8cm}- \jt{-1}{\check{y}_{k+1}}{y_{k+1}}\, \jp{}{\check{y}_{k+1}}{y_{k+1}}\, C_{\xiop}\,\varepsilon_{k+1}.
\end{align}
Following the methodology in Sec.~\ref{sec:update}, substituting \eqref{eq:prior_linearisation} and \eqref{eq:measurement_update_taylor2_iterate} into \eqref{eq:update_goal} and solving for $\Sigma_{k+1|k}^+$ yields
\begin{align}\label{eq:update_sigma_iterate}
    \Sigma_{k+1|k}^+ = \left(\Sigma_\dagger^{-1} + C_{\xiop}^\top R_\dagger^{-1} C_{\xiop} \right)^{-1},
\end{align}
where
\begin{align*}
    &R_\dagger = \jp{-1}{\hat{y}_{k+1}}{y_{k+1}} \jt{}{\hat{y}_{k+1}}{y_{k+1}} \\
    &\hspace{3cm}R_{k+1} \jt{\top}{\hat{y}_{k+1}}{y_{k+1}}\jp{-\top}{\hat{y}_{k+1}}{y_{k+1}}.
\end{align*}
Using \eqref{eq:update_sigma_iterate}, one can solve for the mean $\mu_{k+1|k}^+$ by matching terms in \eqref{eq:update_goal}, which yields
\begin{align}\label{eq:update_mu_iterate}
    &\mu_{k+1|k}^+ = -\Sigma_{k+1|k}^+ \bigg(
     C_{\xiop}^\top \jp{\top}{\check{y}_{k+1}}{y_{k+1}} \jt{\top}{\check{y}_{k+1}}{y_{k+1}}  \nonumber \\
    & \hspace{2.5cm}R_{k+1}^{-1} \varphi_{\check{y}_{k+1}}(y_{k+1})+ \Sigma_\dagger^{-1} \vartheta_{\xiop}(\hat{\xi}_{k+1|k})
    \bigg).
\end{align}

\subsection{Geometric Iterated EKF}\label{sec:geometric_itekf}

The key idea in the iterated EKF is to use a linearisation point  $\xiop\in\calM$ that is as close to the true state $\xi_{k+1}$ as possible.
This linearisation point must chosen by using some additional information; that is, information in addition to the prior $\hat{\xi}_{k+1|k}$ that already summarises the information in all past measurements.
The obvious choice is to use the new measurement $y_{k+1}$ itself.
For example, given the prior $\hat{\xi}_{k+1|k}$ and the measurement $y_{k+1}$, the update step outlined in Section~\ref{sec:update} generates a new state estimate $\hat{\xi}_{k+1|k+1}^+:=\vartheta^{-1}_{\xiop}(\mu^+_{k+1|k})$ (before reset) that should be a better approximation of the true state $\xi_{k+1}$ than the prior.
If we forget for the moment how $\hat{\xi}_{k+1|k+1}^+$ was obtained, then it makes sense to use $\xiop = \hat{\xi}_{k+1|k+1}^+$ as a linearisation point and recompute the update.
Note that this is not introducing stochastic dependence in the filter, since the
linearisation point is not fused in the filter, it is used only to reduce linearisation error, and additionally in the geometric EKF, improve conditioning of the normal coordinates used.
If improving the linearisation point in turn leads to a better estimate of the state, then the process can be repeated, leading to the iterated EKF algorithm, see Algorithm~\ref{alg:iterate_EKF}.

\begin{algorithm}[h]

    \caption{Geometric Iterated EKF on manifold}
    \KwIn{Initial state estimate $\hat{\xi}_0$ and covariance $\Sigma_0$.}
    \KwOut{Sequence of state and covariance estimates $\{ \hat{\xi}_k,\Sigma_k \}$.}
    \For{$k = 0, 1, 2, \ldots$}
    {
        \vspace{2mm}
        \textbf{\{Propagation step\}}\\
        \vspace{1mm}
        Define process noise covariance $Q_{k+1}$\\
        \vspace{2mm}
        $A_{k+1}\leftarrow\tD \vartheta_{\hat{\xi}_{k+1|k}}(\hat{\xi}_{k+1|k}) \cdot \tD F_{u_{k+1}}(\hat{\xi}_{k|k})\cdot \tD \vartheta^{-1}_{\hat{\xi}_{k|k}}(0)$\\
        $\hat{\xi}_{k+1|k}\leftarrow F(\hat{\xi}_{k|k}, u_{k+1})$\\
        $\Sigma_{k+1|k} \leftarrow A_{k+1}\Sigma_{k|k}A_{k+1}^\top + Q_{k+1}$\\
        \vspace{2mm}
        \If{Measurement $y_{k+1}$ available}
        {
            \vspace{2mm}
            \textbf{\{Iterated update step\}}\\
            \vspace{1mm}
            Define measurement noise covariance $R_{k+1}$\\
            Define maximal number of iterations $N_{k+1}$\\
            $i\leftarrow 0$\\
            $\xiop^0\leftarrow \hat{\xi}_{k+1|k}$\\
            \vspace{2mm}
            \While{$i < N_{k+1}$}
            {
                \vspace{2mm}
                $C_{\xiop^i}\leftarrow \tD h(\xiop^i)$\\
                \vspace{2mm}Compute the maps $\jt{}{\hat{\xi}_{k+1|k}}{\xiop^i}$, $\jt{}{\check{y}_{k+1}^i}{y_{k+1}}$, and $\jp{}{\check{y}_{k+1}^i}{y_{k+1}}$\\
                \vspace{2mm}Compute the step $\mu_{k+1|k}^i$ with \eqref{eq:update_mu_iterate}\\
                $\xiop^{i+1}\leftarrow\vartheta^{-1}_{\xiop^i}(\mu_{k+1|k}^i)$\\
                $i\leftarrow i+1$\\
            }
            \vspace{1mm}
            $\xiop\leftarrow \xiop^i$\\
            Compute $\Sigma_{k+1|k}^+$ with \eqref{eq:update_sigma_iterate}\\
            Compute $\mu_{k+1|k}^+$ with \eqref{eq:update_mu_iterate}\\
            \vspace{2mm}
            \textbf{\{Reset step\}}\\
            \vspace{1mm}
            $\hat{\xi}_{k+1|k+1} \leftarrow \vartheta^{-1}_{\xiop}(\mu_{k+1|k}^+)$\\
            Compute the Jacobian $\jt{}{\xiop}{\hat{\xi}_{k+1|k+1}}$\\
            $\Sigma_{k+1|k+1} \leftarrow \jt{}{\xiop}{\hat{\xi}_{k+1|k+1}}\Sigma_{k+1|k}^+ \jt{\top}{\xiop}{\hat{\xi}_{k+1|k+1}}$\\
        }
    }
    \label{alg:iterate_EKF}
\end{algorithm}

    \begin{remark}
        Applying an analogous derivation to that given in \cite{bell1993iterated}, adapted to the geometric setting, it is straightforward to show that the limit point of the iterated update step in Algorithm~\ref{alg:iterate_EKF} computes the the Maximum a Posteriori (MAP) estimate of the inference problem associated with the data fusion.
        This property depends on the Jacobians and proper handling of the geometry and does not hold in general.
     \end{remark}

\section{Case Study: Extended Pose Group}\label{sec:simulation}
In this section we evaluate the performance of the proposed EKF methodology on a simplified inertial navigation system (INS) problem on the extended pose group $\SE_2(3)$.

\subsection{System definition}
Consider estimating the pose and velocity of a robotic vehicle moving relative to a fixed reference.
For example, an aerial vehicle moving relative to a frame fixed to the earths surface, ignoring the earths rotation.
The state of the robot is represented by the rigid body orientation, velocity and position in the global inertial frame of reference, denoted by $\bR\in\SO(3)$, $\bv\in\R^3$ and $\bp\in\R^3$ respectively.
The state space is the extended pose matrix Lie-group $\SE_2(3)$, where we write an element of the state space as
\[
\xi:=\begin{bmatrix}
    \bR & \bv & \bp\\
   \mathbf{0}_{1\times 3} & 1 & 0\\
   \mathbf{0}_{1\times 3} & 0 & 1
\end{bmatrix}\in\SE_2(3).
\]

The robot is equipped with an IMU-type sensor providing bias-free angular velocity $\omega\in\R^3$ and acceleration measurements $a\in\R^3$ at 200 Hz, both corrupted with additive Gaussian noise $\kappa_\omega\sim \GP(0_{3\times 1}, Q_\omega)$ and $\kappa_a\sim\GP(0_{3\times 1},Q_a)$.
In the real-world, all such sensors would be corrupted by bias offsets, however, the focus of the present simulation is to demonstrate the relative performance of the geometric EKF not solve the INS problem, and we make this assumption to avoid complexity in the simulation study.

Define the input matrix $\mathbf{V}$, the gravity matrix $\mathbf{G}$, and $\mathbf{N}$ used to model the linear kinematics $\mathbf{v} = \dot{\mathbf{p}} $, by
\begin{align*}
\mathbf{V} & = \begin{bmatrix}
    \omega^\wedge & a & \mathbf{0}_{1\times 3}\\
    \mathbf{0}_{2\times 3}&\mathbf{0}_{2\times 1}&\mathbf{0}_{2\times 1}
\end{bmatrix}, \\
\mathbf{G} & = \begin{bmatrix}
    \mathbf{0}_{3\times 3} & g & \mathbf{0}_{1\times 3}\\
    \mathbf{0}_{2\times 3}&\mathbf{0}_{2\times 1}&\mathbf{0}_{2\times 1}
\end{bmatrix}, \\
\mathbf{N} & =\left[\begin{array}{ccc}
    \mathbf{0}_{3 \times 3} & \mathbf{0}_{3 \times 1} & \mathbf{0}_{3 \times 1} \\
    \mathbf{0}_{1 \times 3} & 0 & 1 \\
    \mathbf{0}_{1 \times 3} & 0 & 0
    \end{array}\right],
\end{align*}
where $g\in\R^3$ is the gravity vector.
The noise-free discrete-time system dynamics can be written:
\begin{align}\label{eq:sim_system_func}
    \xi_{k+1} = \exp(\delta t(\mathbf{G}-\mathbf{N}))\;\xi_k\;\exp(\delta t(\mathbf{V}+\mathbf{N})),
\end{align}
where $\delta t\in\R^+$ is the time interval, $\exp$ is the matrix exponential.

In a real-world problem the state measurements are typically position from a Global Navigation Satellite System (GNSS), possibly a GNSS velocity, magnetomoeters, barometric pressure,  etc.
Since our goal is to demonstrate relative performance of the algorithm rather than solve the INS filtering problem, we simplify the output model and assume that the pose $(\bR,\bp)$ of the robot in the global reference frame is measured at 10 Hz.
The output space $\calN$ is the pose group $\SE(3)$.
For this simplified output model, the measurement function is
\begin{align}\label{eq:sim_meas_func}
    y_k &= h(\xi_{k})\boxplus\nu_k\nonumber\\
        &= \exp_{\SE(3)}(\nu_k)\begin{bmatrix}
            \bR_k & \bp_k\\
            \mathbf{0}_{1 \times 3} & 1
        \end{bmatrix}, \quad\nu_k\sim\GP(0_{6\times 1},R_k).
\end{align}
The measurement is corrupted with noise in the exponential coordinate, assumed to be right-invariant.

\subsection{Implementation}
In this simulation, we conduct a Monte-Carlo simulation, including 1000 runs of a simulated mobile robot moving in a Lissajous trajectory of 60 seconds length.
The IMU noise is randomly generated following Gaussian distributions, with standard deviation $0.001\text{rad/s}\sqrt{\text{s}}$ for the angular velocity, and $0.01 \text{m/s}^2\sqrt{\text{s}}$ for the acceleration.
The pose measurement noise follows a Gaussian distribution with standard deviation $\diag\{0.4, 0.3, 0.2, 2.0, 1.0, 0.2\}$.
Note that the noise in the position measurement (last three entries) is highly inhomogeneous.
This choice emphasises the role of the geometric corrections in the simulations, since poorly corrected terms in the covariance update will cause more significant performance degradation than they would if the underlying noise distribution was homogeneous.

For comparison, we implement the original EKF as described in \cite{ge2023note}, the geometric EKF and the geometric ItEKF on the same system.
To evaluate the performance of the filters, we compare the estimation error in all three states and the ANEES \cite{li2011evaluation} of the filters.
The ANEES is defined as
\[
{\text{ANEES} = \frac{1}{nM}\sum_{i = 1}^{M}\varepsilon_{i}^T\mathbf{\Sigma}_{i}^{-1}\varepsilon_{i}},
\]
where ${\varepsilon}$ is the local state error, $\mathbf{\Sigma}$ is the error covariance, ${M = 1000}$ is the number of runs in the Monte-Carlo simulation, and $n=9$ is the dimension of the state space.

The ANEES provides a measure of the consistency of the filter estimate.
For a stochastically consistent filter, with no linearisation error and for perfectly Gaussian data, the ANEES should follow a $\chi$-squared distribution with expected value unity.
When the ANEES is larger than unity, it indicates that the filter is overconfident; that is, the observed error is larger than the estimate of the state covariance predicted.
This is usually due to linearisation error that is not modelled in the noise process but contributes to the observed error.
Conversely, if the ANEES is less than unity the observed error is smaller than predicted and the filter is under-confident.
If the noise processes are correctly modelled, it is expected that an EKF should be overconfident (since linearisation error will be percieved as unmodelled noise).
In real-world scenarios, an engineer will often overestimate the noise covariance for measurement processes to compensate for linearisation error and avoid overconfidence of the filter.
Such real-world considerations are not in the scope of this paper, and we focus on demonstrating the underlying performance properties of the filters considered.


\subsection{Results}
\begin{figure}[htb!]
    \centering
    \includegraphics[width=0.5\linewidth]{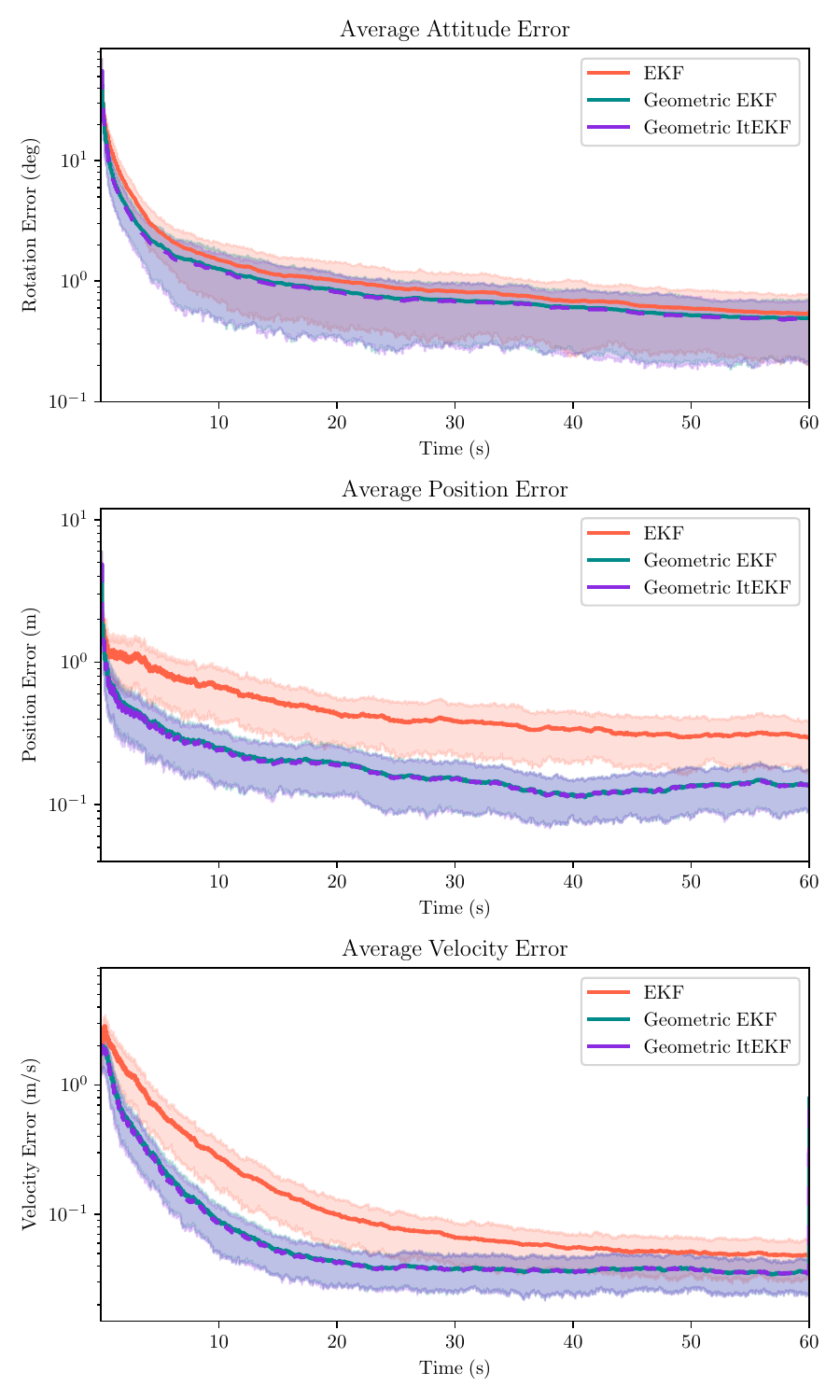}
    \caption{The estimation error are shown for different EKF implementations. The \textcolor{orange}{original EKF} is compared with the \textcolor{OliveGreen}{geometric EKF} and \textcolor{Purple}{ItEKF} proposed in this work. The shaded area represents the 25th and 75th percentiles of the estimation error.}
    \label{fig:result_error}
\end{figure}

\begin{figure}
    \centering
    \includegraphics[width=0.5\linewidth]{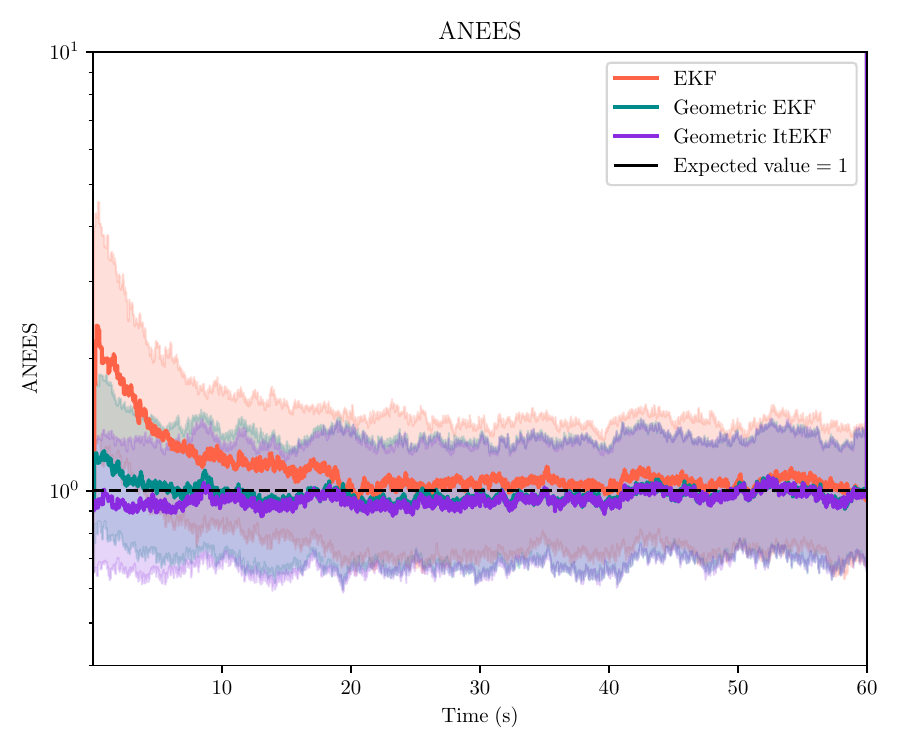}
    \caption{The ANEES of the different EKF implementations (\textcolor{orange}{original EKF}, \textcolor{OliveGreen}{geometric EKF} and \textcolor{Purple}{geometric ItEKF}) are shown. The \textcolor{black}{black dashed line} represents the ideal ANEES value of 1.}
    \label{fig:result_anees}
\end{figure}
\begin{table}[h]
    \centering
    \renewcommand{\arraystretch}{1.5}
    \begin{subtable}{\columnwidth}
        \centering
        \caption{RMSE in the transient Phase (0-30s)}
        \resizebox{0.5\columnwidth}{!}{ 
        \begin{tabular}{lccc}
            \toprule
            & \multicolumn{1}{c}{Rot. RMSE (deg)} & \multicolumn{1}{c}{Pos. RMSE (m)} & \multicolumn{1}{c}{Vel. RMSE (m/s)} \\
            \midrule
            EKF    & 5.0782 (100\%) & 0.7226 (100\%) & 0.6271 (100\%) \\
            Geometric EKF   & \underline{4.5332 (89.2\%)} & \underline{0.4168 (57.7\%)} & \underline{0.3602 (57.4\%)} \\
            \textbf{Geometric ItEKF} & \textbf{4.4634 (87.9\%)} & \textbf{0.4103 (56.8\%)} & \textbf{0.3493 (55.7\%)}\\

            \midrule

            ItEKF & 5.0642 (99.7\%) & 0.7146 (98.9\%) & 0.6108 (97.4\%) \\
            Geometric EKF (Update only) & 5.0243 (98.9\%) & 0.4804 (66.5\%) & 0.4868 (77.6\%) \\
            Geometric EKF (Reset only) & 8.8802 (174.8\%) & 1.6031 (221.9\%) & 1.5285 (243.8\%) \\
            Geometric ItEKF (Update only) & \textbf{4.4634 (87.9\%)} & \textbf{0.4103 (56.8\%)} & \textbf{0.3493 (55.7\%)} \\
            \bottomrule
        \end{tabular}
        }
    \end{subtable}

    \vspace{5pt} 

    \begin{subtable}{\columnwidth}
        \centering
        \caption{RMSE in the asymptotic Phase (30-60s)}
        \resizebox{0.5\columnwidth}{!}{ 
        \begin{tabular}{lccc}
            \toprule
            & \multicolumn{1}{c}{Rot. RMSE (deg)} & \multicolumn{1}{c}{Pos. RMSE (m)} & \multicolumn{1}{c}{Vel. RMSE (m/s)} \\
            \midrule
            EKF    & 0.6615 (100\%) & 0.3301 (100\%) & 0.0553 (100\%) \\
            Geometric EKF   & 0.5790 (87.5\%) & \underline{0.1327 (40.2\%)} & \underline{0.0382 (69.1\%)} \\
            \textbf{Geometric ItEKF} & \textbf{0.5673 (85.8\%)} & \textbf{0.1319 (40.0\%)} & \textbf{0.0381 (68.9\%)} \\
            \midrule
            ItEKF & 0.6616 (100.0\%) & 0.3288 (99.6\%) & 0.0551 (99.6\%) \\
            Geometric EKF (Update only) & 0.6739 (101.9\%) & 0.1391 (42.1\%) & 0.0393 (71.1\%) \\
            Geometric EKF (Reset only) & 1.5221 (230.1\%) & 0.4500 (136.4\%) & 0.0805 (145.6\%) \\
            Geometric ItEKF (Update only) & \underline{0.5674 (85.8\%)} & \textbf{0.1319 (40.0\%)} & \textbf{0.0381 (68.9\%)} \\
            \bottomrule
        \end{tabular}
        }
    \end{subtable}
    \vspace{5pt} 
    \caption{RMSE Comparison for different EKF implementations in the transient and asymptotic phases. The percentage indicates the improvement of the geometric filters with respect to the original EKF. \textbf{Bold} values indicate the best performance. \underline{Underlined} values indicate the second best performance.}
    \label{tab:rmse_comparison}
\end{table}

Figure~\ref{fig:result_error} shows the performance of the different EKF implementations, with and without the proposed geometric modifications.
To make the differences more visible, Table~\ref{tab:rmse_comparison} shows the average RMSE values of each filter in the transient (first 30s) and asymptotic (30-60s) phases of the trajectory.
The proposed geometric filters are seen to make a noticeable improvement to the filter performance throughout the trajectory, in both the transient and asymptotic phases.
Comparing to the original EKF, the geometric filters have much lower estimation error, especially in the position and velocity components.
The geometric correction terms are important in correctly aligning the measurement information when it is incorporated into the covariance estimate during the update step and the error in the position and velocity are emphasised in this example by the choice of highly inhomogeneous noise in the position measurement.
The geometric ItEKF further improves the estimation accuracy compared to the geometric EKF (cf.~Table~\ref{tab:rmse_comparison}), however, the advantage is less significant than may have been expected.
This is partly due to the simple output function \eqref{eq:sim_meas_func} that is linear in the state.
As a consequence there is no benefit in relinearisation of the output function and the geometric iterated EKF terminates on the first iteration.
There is still a slight advantage in the geometric ItEKF since the update is now undertaken with respect to a reference point $\xiop\in\calM$ that should be closer to the true state than $\hat{\xi}_{k+1|k}$, and the geometric correction terms are slightly more effective.
This is particularly true during the transient phase of the filter when there are large state errors.
Conversely, during the asymptotic phase, the geometric ItEKF and geometric EKF show the same performance.

Figure~\ref{fig:result_anees} shows the ANEES of the different EKF implementations.
The original EKF is significantly overconfident in the transient phase due to unmodelled linearisation error.
The geometric EKF improves the consistency of the filter, roughly halving the overconfidence demonstrating that a significant component of the linearisation error is due to the geometric structure of the space, not nonlinearity of the measurement function.
This should be expected in an example such as this, where the measurement function is linear as discussed earlier.
The geometric ItEKF further improves the performance, particularly in the transient phase, due to the better choice of reference point in the update.
As shown in Figure~\ref{fig:result_anees} and Table~\ref{tab:rmse_comparison}, although both geometric EKF and ItEKF have very similar RMSE values, the geometric ItEKF is more consistent in the transient phase, which is reflected in the ANEES values.

\subsection{Ablation Study}
In this section, we conduct an ablation study to evaluate the impact and importance of the `update' and `reset' geometric modifications in the EKF methodology.

In Figure~\ref{fig:ablation_error} (left hand column) and Figure~\ref{fig:ablation_anees} (first plot) we implement the geometric EKF with only the update modification, with only the reset modification, and with both modifications.
For comparison we also implement the standard EKF.

Interestingly, the geometric EKF with only the reset modification significantly degrades the estimation accuracy and consistency compared to the original EKF.
The geometric EKF with only the update modification shows improved performance compared to the basic EKF, however, it is only when both geometric modificaitons are combined that the full performance gain is realised.
Figure~\ref{fig:ablation_anees}, emphasises this performance difference, where the geometric EKF with only the update modification is even more overconfident than the original EKF.
This result indicates that the main benefit in considering the geoemtric perspective is gained in the update step, not the reset.
However, Figure~\ref{fig:ablation_anees} also demonstrates that leaving out the reset step leads to poor consistency of the filter, with the basic EKF outperforming both the partial geometric filters in ANEES.
In conclusion, both geometric modifications are important in obtaining high performance filters and they work synergistically not independently.

\begin{figure}[htb!]
    \centering
    \includegraphics[width=0.8\linewidth]{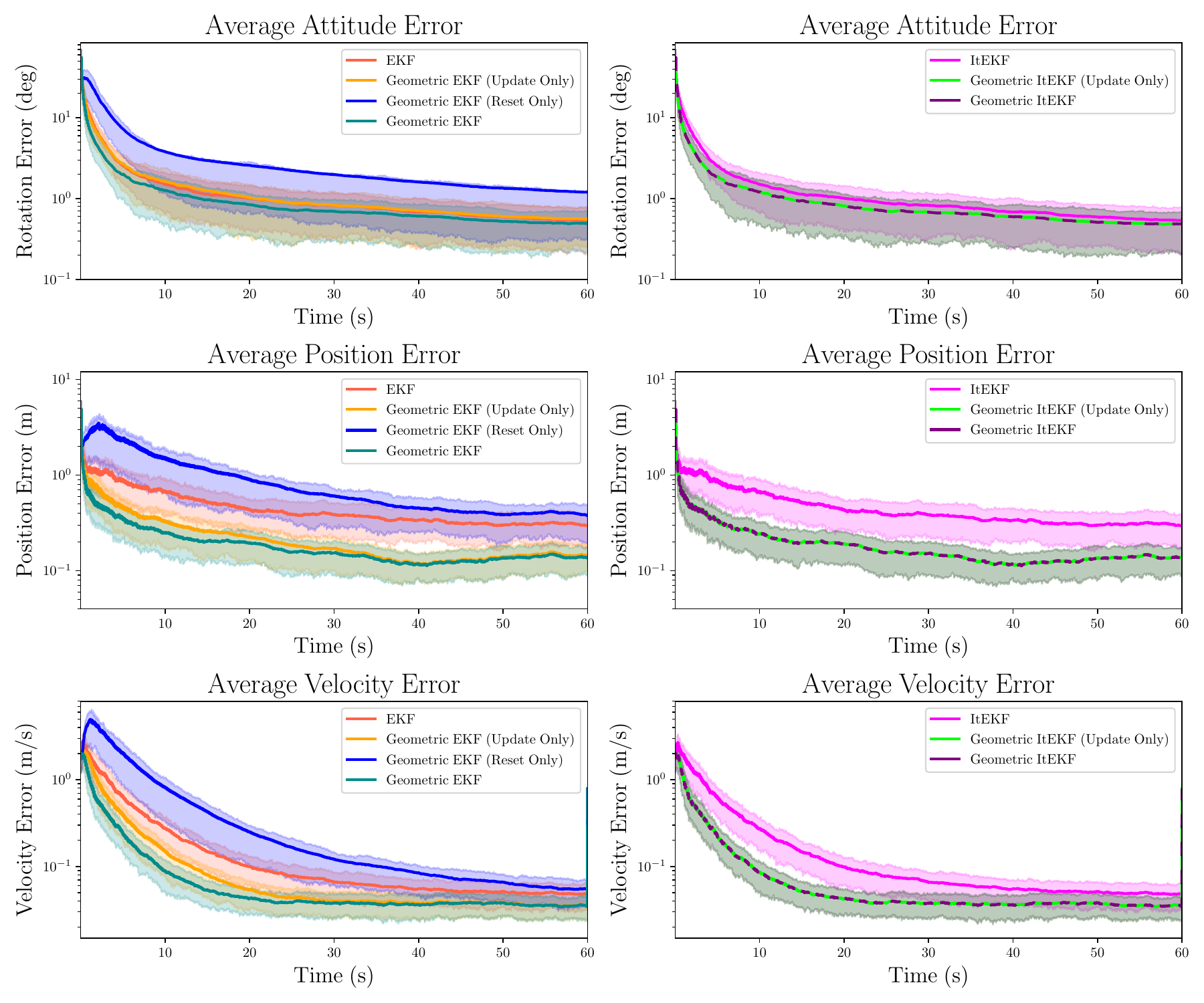}
    \caption{The estimation error is shown for different EKF implementations. On the left the non-iterated algorithms: The \textcolor{OrangeRed}{EKF}, the \textcolor{orange}{geometric EKF (update only)}, the \textcolor{blue}{geometric EKF (reset only)}, the \textcolor{OliveGreen}{full geometric EKF}.
    On the right the iterated algorithms: The \textcolor{Mulberry}{iterated EKF}, the
    \textcolor{LimeGreen}{geometric ItEKF (update only)}, and the the \textcolor{Purple}{full geometric EKF}.
    There is no reset in an iterated EKF as the update is computed using the new state as the reference point.
    The shaded area represents the 25th and 75th percentiles of the estimation error.}
    \label{fig:ablation_error}
\end{figure}

\begin{figure}[htb!]
    \centering
    \includegraphics[width=0.5\linewidth]{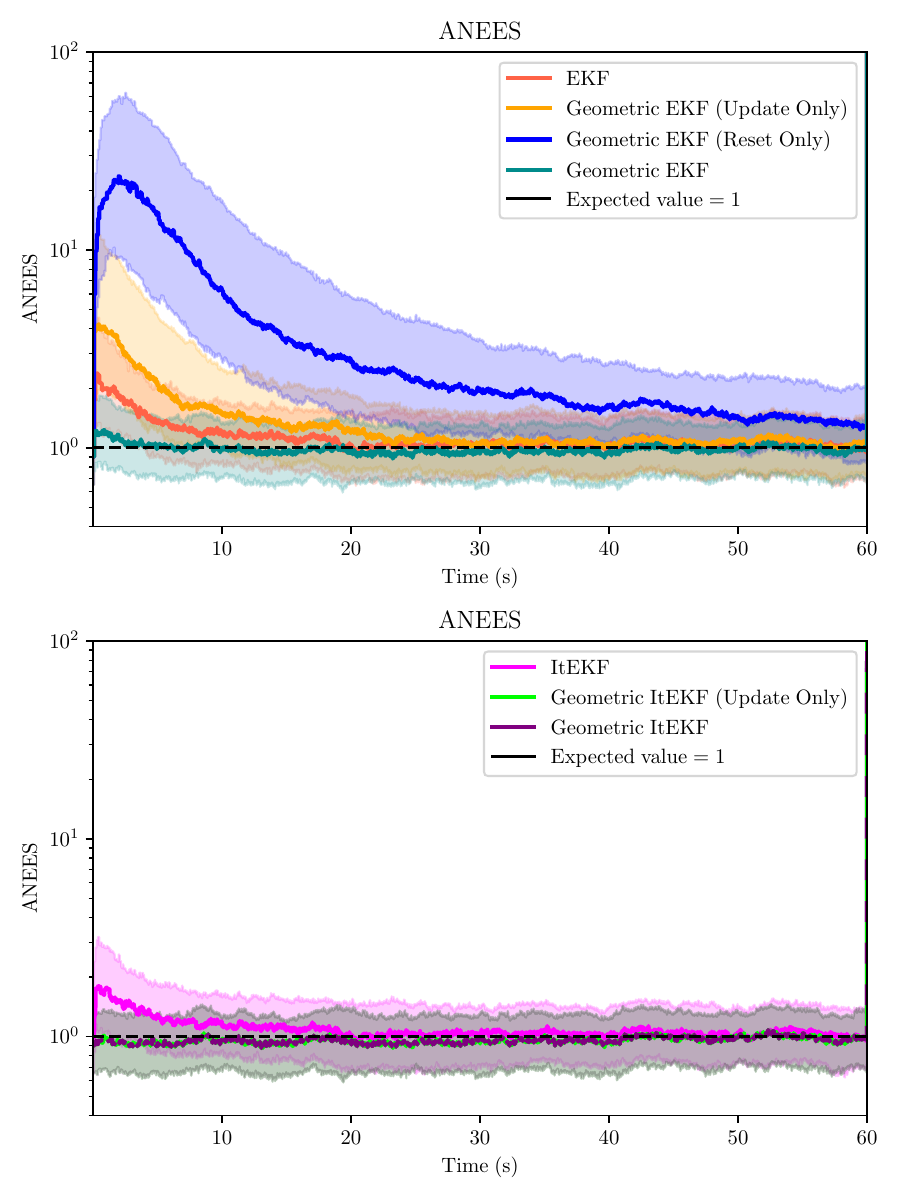}
    \caption{The ANEES of the different EKF implementations are shown. The black dashed line represents the ideal ANEES value of 1. }
    \label{fig:ablation_anees}
\end{figure}

In Figure~\ref{fig:ablation_error} (right hand column) and Figure~\ref{fig:ablation_anees} (second plot) we implement the iterated EKF, geometric iterated EKF with only the update modification and the full geometric iterated EKF.
The iterated EKF (ItEKF) is an implmentation of the ItEKF in \cite{bell1993iterated} computed in exponential coordinates.

For the iterated filters, interestingly, the geometric ItEKF with only the update modification provides almost identical performance to the full geometric ItEKF.
Note that including the reset step makes little difference to an iterated EKF since the reference point for the update is already the best estimate of the state estimate once the iteration has converged.
Hence the final update computation is already expressed in concentrated Gaussian around the final linearisation point and the reset does nothing.
In practice, the iteration is often terminated before the full convergence is obtained, however, the residual reset is negligable (see also \cite{ge2022equivariant})
The simulations clearly show that the geometric iterated EKF with only the update modification and the full geometric iterated EKF demonstrate the same performance in both error and ANEES.
Both algorithms outperform a classical iterated EKF implmentation.

The ablation study shows that both the update and reset modifications are essential to improve the performance of the EKF, while the reset step becomes less significant in the iterated EKF.

\section{Conclusion}
In this work, we exploit the geometric structure of the system spaces and present a concentrated Gaussian distribution model on smooth manifolds with affine connections.
An error-state extended Kalman filter design methodology is presented based on this model, which includes additional geometric modifications in the filter dynamics to model the coordinate change and curvature in the space.
The theory is applied to an example problem on the extended pose group $\SE_2(3)$ with pose measurements.
The simulation results show that the proposed methodology outperforms the standard EKF and the iterated EKF in terms of estimation accuracy and consistancy.

\appendix
\section{Appendix}
\subsection{Approximating partial pushforward maps of the exponential mapping}\label{sec:approximation_dexp}

The results in Lemma~\ref{lemma:standard_linearisation} as well as the filter update and reset steps rely on evaluating the partial pushforward maps of the exponential mapping \cite{lee2018introduction}.
The tangential partial push forward, where the base point of the exponential is fixed and the differential is computed with respect to varying the argument, is termed the Jacobian $\jt{}{\xi_1}{\xi_2}$ in the robotics literature.
For certain specific cases, such as the Lie-groups $\SO(3)$ and $\SE(3)$, the Jacobian can be written as an algebraic expression depending on common trigonometric functions \cite{Chirikjian_2011}.
These results hide the reality that the formula for this map on a general geometric manifold is a transcendental function that can only be numerically computed through an infinite series expansion \cite{helgason1979differential}.
Barfoot \etal \cite{barfoot2017state} proposes a numerical Taylor series from the exponential function on Lie-groups.
Chirikjian \cite{Chirikjian_2011} provides a detailed discussion of the Jacobian of the exponential and discusses the Campbell-Baker-Hausdorff formulae and how this can be used in filter updates \cite{wolfe2011bayesian}.
In \cite{ge2024geometric}, the authors propose approximating the Jacobian map on a Lie-group using parallel transport and curvature of the Cartan-Schouten symmetric (0)-connection.
In this section, we show that the same approach can be applied to geometric manifolds, and moreover, can be extended to approximate the positional partial pushforward map $\jp{}{\xi_1}{\xi_2}$,
where the argument of the exponential is fixed and the differential is computed with respect to varying the base point.

\subsubsection{Tangential partial pushforward map or Jacobian}
In the following section we distinguish strongly between the Jacobi field $\calJ_w(t)$ written in caligraphic script, and the Jacobian $\jt{}{\xi_1}{\xi_2}$ written in italic script.
Although they are closely related, the subscripts and arguments are different.
\begin{lemma}\label{lemma:ptcv}
    Let $\calM$ be a geometric manifold with symmetric affine connection $\nabla$.
    Given $\xi_1\in\calM$ and $\xi_2\in U_{\xi_1}$ with $\xi_2 = \vartheta_{\xi_1}(u)$, let $\gamma_v(t):=\vartheta_{\xi_1}(tv),\; t\in [0,1]$ denote the geodesic emanating from $\xi_1$ to $\xi_2$.
    Then the Jacobian map $\jt{}{\xi_1}{\xi_2}:\tT_{\xi_1}\calM\to\tT_{\xi_2}\calM$ satisfies
    \begin{align}\label{eq:ptcv}
        \jt{}{\xi_1}{\xi_2}[w] = \PT_{\gamma_v}^{0,1} \left(w +\frac{1}{6}\calR(v,w)v\right)+ \order(\lvert v \rvert^3),
    \end{align}
    for $w\in\tT_{\xi_1}\calM$.
\end{lemma}

The little `o'-order notation $\order(\lvert v \rvert^3)$ denotes that
    \[
    \lim_{v \rightarrow 0}\frac{1}{\lvert v \rvert^3} \left( \jt{}{\xi}{\gamma_v(1)}[w] - \PT_{\gamma_v}^{0,1}\left(w + \frac{1}{6}R(v,w)v \right) \right)  = 0
    \]
and is a coordinate independent property on $\tT_\xi \calM$.

\begin{proof}
    The Jacobi field $\mathcal{J}_w(t)$ is a vector field along $\gamma_v(t)$ with $J_w(t)\in\tT_{\gamma_v(t)}\calM$ defined as the solution of the Jacobi equation \cite{lee2018introduction}
    \begin{gather}
    \tD_t^2 \mathcal{J}_w(t)  = \calR(\dot{\gamma}_v(t), \mathcal{J}_w(t) ) \dot{\gamma}_v(t),
    \label{eq:Jacobi_ODE}
    \end{gather}
    where $\tD_t = \nabla_{\dot{\gamma}_v}$ is the covariant derivative along the curve $\gamma_v(t)$.
    Choosing initial conditions $\mathcal{J}_0(w) = 0$, and $\tD_t\mathcal{J}_0(w)= w$, then for $t\neq 0$, the Jacobian and Jacobi field satisfy \cite[Theorem 3.1]{lang2012fundamentals}
    \[
    \jt{}{\xi_1}{\gamma_v(t)}[w] = \frac{1}{t} \mathcal{J}_t(w).
    \]
    In order to study the Taylor series of the Jacobian, one needs to first identify $\tT_{\gamma_v(t)}\calM$ with $\tT_{\xi_1}\calM$ using parallel transport along $\gamma_v(t)$ in the reverse direction.
    Construct the following map
    \begin{align}\label{eq:jacobi_map}
        v \mapsto (\PT_{\gamma_v}^{0,1})^{-1}\left(\jt{}{\xi_1}{\gamma_v(t)}[w]\right),
    \end{align}
    which maps from $\tT_{\xi_1}\calM$ to $\tT_{\xi_1}\calM$.
    Studying the partial derivative of \eqref{eq:jacobi_map} in the $v$-variable is equivalent to taking the derivative of
    \[
    t\to(\PT_{\gamma}^{0,t})^{-1}\left(\frac{1}{t}\mathcal{J}_t(w)\right).
    \]
    at $t=0$.
    Applying \cite[Theorem A.2.9]{waldmann2012geometric} yields
    \begin{align*}
        \jt{}{\xi_1}{\gamma_v(1)}[w] = \PT_\gamma^{0,1} \left(w +\frac{1}{6}\calR(v,w)v\right)+ \order(\lvert v \rvert^3).
    \end{align*}
\end{proof}

\begin{remark}
Note that in \eqref{eq:ptcv}, both $\PT_{\gamma_v}^{0,1}(w)$ and $\calR(v,w)v$ are linear in $w$, hence the whole object $\PT_{\gamma_v}^{0,1}(w + \frac{1}{6}\calR(v,w)v)$ is a linear map of $w$, which makes this result compatible with Lemma~\ref*{lemma:standard_linearisation}.
\end{remark}

Lemma \ref{lemma:ptcv} provides a computationally feasible approach to approximate the Jacobian.
Since the error in \eqref{eq:ptcv} is $\order(\lvert v\lvert^3)$ then the third order term in the expansion is zero and the Taylors approximation is fourth order.
If computing the Riemannian curvature is numerically challenging then
    \begin{align}\label{eq:pt}
        \jt{}{\xi_1}{\xi_2}[w] = \PT_\gamma^{0,1} \left(w \right)+ \order(\lvert v \rvert),
    \end{align}
provides a second order Taylors approximation.
Of course, the parallel transport function itself may be difficult to compute and require its own numerical approximation.

On a Lie-group, the exponential is defined independently of the affine connection.
In order to apply Lemma~\ref{lemma:ptcv} one must choose the Cartan-Schouten (0)-connection; the unique symmetric affine connection who's geodesics are the 1-parameter subgroups or exponentials on the Lie-group.
It is of interest to consider the case of the (-) and (+) connections, who's geodesics also correspond to exponentials on the Lie-group.
These connections are not symmetric, that is they have non-zero torsion, although they are flat, that is they have zero curvature.
The terms in the Jacobi field expansion will be different, although the two infinite asymptotic expansions must converge in the limit, since the different expansions compute the same function.
Curtailing the Jacobi field expansion at the quadratic term for different affine connections leads to different approximations of the Jacobian.
In general, the fact that the third order term in the Jacobi field expansion is zero for a symmetric affine connection (and the effective Taylor series approximation is 4th order) leads to a better numerics than for a non-symmetric affine connection, as long as the advantage is not lost in the computing the parallel transport.
The parallel transport for the (0)-connection can be computed easily \cite{ge2024geometric}.

\subsubsection{Positional partial pushforward map}
The positional partial pushforward map $\jp{}{\xi_1}{\xi_2}$ can be approximated in a similar manner.

\begin{lemma}
Let $\calM$ be a geometric manifold with symmetric affine connection $\nabla$.
Given $\xi_1\in\calM$ and $\xi_2\in U_{\xi_1}$ with $\xi_2 = \vartheta_{\xi_1}(v)$, let $\gamma_v(t):=\vartheta_{\xi_1}(tv),\; t\in [0,1]$ denote the geodesic emanating from $\xi_1$ to $\xi_2$.
Then the partial pushforward map of the exponential map with respect to the positional component $\jp{}{\xi_1}{\xi_2}:\tT_{\xi_1}\calM\to\tT_{\xi_2}\calM$ satisfies
\begin{align}\label{eq:positional_ptcv}
\jp{}{\xi_1}{\xi_2}[w] = \PT_{\gamma_v}^{0,1} \left(w -\frac{1}{2}\calR(v,w)v\right)+ \order(\lvert v \rvert^3),
\end{align}
for $w\in\tT_{\xi_1}\calM$.
\end{lemma}

\begin{proof}
The proof constructs a family of geodesics and then applies Jacobi field theory.

Choose a smooth curve $x(s)$ in $\calM$ passing through $\xi_1$ with velocity $w$:
\[
x(0) = \xi_1, \quad \dot{x}(0) = w.
\]
Define a vector field $v(s)$ along $x(s)$ by the parallel transport of $v$.
That is
\begin{align}\label{eq:parallel_vector_field}
    \nabla_{\dot{x}(s)}v(s) = 0, \qquad v(0) = v.
\end{align}
Define a family of geodesics
\[
\gamma^s(t) := \vartheta^{-1}_{x(s)}(tv(s)).
\]
For each fixed $s$, the curve $\gamma^s(t)$ is a geodesic emanating from $x(s)$ with initial velocity $v(s)$.
Note that when $s=0$, we recover the original geodesic $\gamma^0(t)=\gamma_v(t) = \vartheta^{-1}_{\xi_1}(tv)$.

The first variation of $\Gamma(s,t)$ with respect to $s$
\[
\calJ_s(t) = \frac{\partial}{\partial s}\gamma^s(t)
\]
is a Jacobi field along the geodesic $\gamma^s(t)$ \cite{lee2018introduction}.
To solve for $\calJ_s(t)$ using \eqref{eq:Jacobi_ODE} one needs to identify appropriate initial conditions at $t=0$.
From $\gamma^s(0) = x(s)$ then
\[
\calJ_0 (0) = \dds\bigg|_{s=0}\gamma^s(0)  = \dds\bigg|_{s=0}x(s) = w.
\]
At $t=0$, then $\dot{\gamma}^s(0) = v(s)$.
Hence one has
\begin{align*}
\tD_t\calJ_0(t) \bigg|_{t=0}
&= \tD_t\bigg|_{t=0} \frac{\partial}{\partial s}\bigg|_{s=0} \gamma^s(t) \\
& = \tD_s \bigg|_{s=0} \frac{\partial}{\partial t} \bigg|_{t=0} \gamma^s(t) \\
& = \tD_s \bigg|_{s=0} v(s) = 0
\end{align*}
where the second line follows since the affine connection is symmetric and the Lie-bracket of coordinate differentials is zero $[\frac{\partial}{\partial t} , \frac{\partial}{\partial s}] = 0$, and the final equality follows since $v(s)$ is a parallel vector field along $x(s)$.
In conclusion, the Jacobi field $\calJ_s(t)$ has the initial conditions
\[
\calJ_0(0) = w,\quad \tD_t \calJ_0(0) = 0.
\]
To compute the positional partial pushforward map, we need to compute the Jacobi field at $t=1$, that is,
\[
\jp{}{\xi_1}{\xi_2}[w] = \calJ_0(1).
\]

A similar approximation scheme as in Lemma~\ref{lemma:ptcv} can be applied to the Jacobi field $\calJ_s(t)$.
Define a vector $K(t)\in\tT_{\xi_1}\calM$ by parallel transporting $\calJ_0(t)$ back to the initial point, that is,
\[
K(t): = {\PT_{\gamma_v}^{0,t}}^{-1}\left(\calJ_s(t)\right) = \PT_{\gamma_v}^{t,0}\left(\calJ_s(t)\right).
\]
One can then rewrite the Jacobi equation in terms of $K(t)$:
\[
\frac{\td^2}{\mathrm{dt}^2} K(t) = -\PT_{\gamma_v}^{t,0} \left(\calR(\dot{\gamma}(t),\PT_{\gamma}^{0,t}(K(t))\dot{\gamma}(t)\right).
\]
as a second order ODE on $\tT_{\xi_1} \calM$, a Euclidean space.
Hence $K$ satisfies a second-order ODE that integrates up the curvature along the path with initial conditions
\[
K(0) = \calJ_0(0) = w,
\qquad
\ddt\bigg|_{t=0} K(t) = 0.
\]
One can then show by Taylor expansion that for small $t$,
\[
K(t) = w - \frac{t^2}{2}\calR(v,w)v + \order(\lvert t \rvert^3),
\]
Assuming $t$ is small is equivalent to assuming $v$ small, the result can be rewritten in terms of $v$ and evaluated at $t=1$
\[
\jp{}{\xi_1}{\xi_2}[w] =\calJ_0(1) =  \PT_{\gamma_v}^{0,1} \left(w -\frac{1}{2}\calR(v,w)v\right)+ \order(\lvert v \rvert^3).
\]
This completes the proof.
\end{proof}

\begin{remark}
    In the proof, we have chosen to fix the vector $v$ by defining a parallel vector field $v(s)$ along the curve $x(s)$, satisfying \eqref{eq:parallel_vector_field}.
    This leads to the initial condition of the Jacobi field $\tD_t \calJ_0(0) = 0$.
    However, one may choose to fix the vector $v$ in a different manner, leading to a general initial condition $\tD_t \calJ_0(0) = \nabla_w v$, following from differentiating $\frac{\partial}{\partial t}\gamma^s(0)=v$ with respect to $s$ at $s=0$..
    In this case, the final result will be
    \[
    \jp{}{\xi_1}{\xi_2}[w] = \PT_{\gamma_v}^{0,1} \left(w + \nabla_w v-\frac{1}{2}\calR(v,w)v\right)+ \order(\lvert v \rvert^3).
    \]
    In this proof, we have chosen $v$ to be unchanged along the geodesic via parallel transport, hence $\nabla_w v = 0$.
\end{remark}

\bibliography{reference}
\bibliographystyle{IEEEtran}

\clearpage

\end{document}